%% file: main.tex
\newif\ifreport
\reporttrue   
\PassOptionsToPackage{draft,bookmarks=false}{hyperref}
\documentclass[conference]{IEEEtran}
\IEEEoverridecommandlockouts
\usepackage{cite}
\usepackage{amsmath,amssymb,amsfonts}
\usepackage{amsmath}
\usepackage{algorithmic}
\usepackage{graphicx}
\usepackage{textcomp}
\usepackage{xcolor}
\usepackage{subcaption}
\usepackage{verbatim}
\usepackage{mathtools}
\usepackage{xurl}
\usepackage[hidelinks]{hyperref}
\newcommand{\amax}{\operatorname*{arg\,max}}
\newcommand{\amin}{\operatorname*{arg\,min}}

\usepackage{color}

\def\red{\color{red}}

\newcommand{\ignore}[1]{}

\newtheorem{assumption}{Assumption}
\newtheorem{proposition}{Proposition}

\newtheorem{lemma}{Lemma}
\newtheorem{theorem}{Theorem}

\def\BibTeX{{\rm B\kern-.05em{\sc i\kern-.025em b}\kern-.08em
    T\kern-.1667em\lower.7ex\hbox{E}\kern-.125emX}}

\makeatletter
\newcommand{\linebreakand}{%
  \end{@IEEEauthorhalign}
  \hfill\mbox{}\par
  \mbox{}\hfill\begin{@IEEEauthorhalign}
}
\makeatother

\begin{document}

\title{





Distributed Cross-Layer Optimization for Covert Multi-Hop, Multi-Modal Networks: Exponentially Fast Convergence and Robust Tracking

\thanks{$^*$Co-primary authors.

This work was supported in part by the ARL cooperative agreement W911NF-24-2-0205, the NSF Grant CNS-2239677, and a hardware donation from NVIDIA.}
}

\author{
\IEEEauthorblockN{Sirin Chakraborty$^{*\dag}$, Andrea Panebianco$^{*\dag}$, Yuchen Tian$^{\dag}$, Kevin S. Chan$^\ddag$,\\ Fikadu Dagefu$^\ddag$, Yin Sun$^{\dag}$, and Ness B. Shroff~\!$^\S$}
{$^\dag$\textit{Dept. ECE, Auburn University, Alabama, USA}}\\
{$^{\ddag}$\textit{DEVCOM Army Research Laboratory, Maryland, USA}}\\
{$^\S$\textit{Depts. ECE and CSE, The Ohio State University, Ohio, USA}}

}
\maketitle

\begin{abstract}
This paper develops the first distributed cross-layer algorithm
for joint congestion control, routing, scheduling, and power
control in covert multi-hop, multi-modal wireless networks, where adversarial wardens  (Willies) monitor radio modalities via energy detection. The Detection Error Probability (DEP), the probability that a Willie fails to reliably detect ongoing transmissions, is generally non-concave in the transmit powers,  making DEP-based covert network optimization challenging. We resolve this by constructing the tightest concave lower bound on the log-DEP, yielding a conservative convex problem that guarantees satisfaction of the original DEP constraints and unifies hard covertness constraints and covertness-utility maximization in a single problem. We develop a Parallel Proximal Alternating Direction Method of Multipliers (PP-ADMM) algorithm for the resulting cross-layer problem and prove global Q-linear convergence, i.e., exponentially fast convergence, to the set of optimal solutions under standard regularity conditions. Numerical results confirm linear convergence and demonstrate robust tracking performance under channel fading and Willie mobility.
\end{abstract}


\vspace{-5mm}
\section{Introduction}
\label{sec:intro}
Wireless tactical networks require reliable and covert information
exchange in contested environments, where adversarial wardens (Willies)
monitor ongoing transmissions. Covert communication offers physical-layer techniques to limit transmission detectability, quantified by the Detection
Error Probability (DEP), the probability that a Willie fails to reliably detect ongoing transmissions~\cite{kong2019distributed,yan2019,kong2022covert}.
In multi-hop, multi-modal networks, where nodes communicate over
heterogeneous radio modalities, covert network control becomes a distributed cross-layer optimization problem. Each Willie monitors one or more radio modalities via energy detection and observes the aggregate received signal energy from all active covert links on each monitored modality. Because the
DEP on each modality depends on the transmit power of all links sharing that modality, the per-Willie covertness constraint couples the power-control
decisions across the links. Moreover, covert power control further interacts
with higher-layer congestion control, routing, and scheduling decisions.

Cross-layer optimization for multi-hop wireless networks has been widely
studied through the Network Utility Maximization (NUM) framework~\cite{eryilmaz2006joint,lin2006utility,neely2010stochastic,wang2017fast},
whereas these NUM formulations do not address DEP-based covertness
constraints. Network-level covert studies address multi-hop routing~\cite{sheikholeslami2018multihop}, joint routing and resource allocation in
heterogeneous networks~\cite{kong2024covert}, multi-flow routing and resource
allocation~\cite{kong2025multiflow}, and decentralized cluster-based
multi-modal routing~\cite{haque2026decor}. However, these studies do not address cross-layer network optimization
under DEP constraints determined by the aggregate received signal energy
from covert links in a multi-hop network. To the
best of our knowledge, this paper presents the first distributed cross-layer
algorithm for joint congestion control, routing, scheduling, and power
control in covert multi-hop, multi-modal wireless networks. 
The technical contributions are  as follows:
\begin{itemize}


    \item The DEP is generally non-concave in the transmit powers, making DEP-based covert network optimization challenging. We resolve this by constructing the tightest concave lower bound on the log-DEP, yielding a conservative convex problem that jointly optimizes congestion control, routing, scheduling, and power control while guaranteeing satisfaction of the original DEP constraints. Moreover, this framework unifies hard DEP constraints and covertness-utility maximization in a single problem.

    \item We develop a Parallel Proximal Alternating Direction Method of Multipliers (PP-ADMM) algorithm for the resulting distributed cross-layer covert network control problem. Under standard regularity conditions, we prove that this distributed algorithm achieves global Q-linear convergence to the set of optimal solutions.


    \item Our numerical results confirm that the PP-ADMM algorithm converges exponentially fast to the optimum and exhibits robust tracking performance under channel fading and Willie mobility.
    

\end{itemize}
For the convergence analysis, the most relevant prior work is~\cite{wang2017fast}, which provides a linearly convergent NUM algorithm
via a two-block sequential ADMM. Our paper makes three advances beyond~\cite{wang2017fast}. First, because our capacity
and covertness constraints are nonlinear, the piecewise-polyhedral
calmness argument used in~\cite{wang2017fast} is not applicable.
Instead, we use a Karush–Kuhn–Tucker (KKT) metric-subregularity error bound and derive boundedness of the primal–dual iterates from sufficient descent.
Second, our proximal regularization enables fully parallel (Jacobi)
primal-block updates, whereas~\cite{wang2017fast} requires sequential (Gauss-Seidel) primal updates. Third, our PP-ADMM solves each local proximal
subproblem directly, without the inexact Uzawa approximation used
in~\cite{wang2017fast}. The resulting multi-block parallel proximal ADMM algorithm makes the convergence analysis substantially more challenging than that in~\cite{wang2017fast}.

\ignore{
\textcolor{blue}{This work makes three technical advances over~\cite{wang2017fast} and the link-level covert communication
literature~\cite{yan2019,kong2022covert}. First, the log-DEP
function is non-concave, and prior work~\cite{yan2019,kong2022covert} analyze it only at
the link level; we construct the highest concave lower bound on the
log-DEP that is tight on its concave region
(Lemma~\ref{lem:log_dep_curvature} and~\eqref{eq:g_definition}),
turning the per-Willie covertness requirement into a convex
constraint for network-level control. Second, our capacity and
covertness constraints are nonlinear in the transmit powers $p$ and bandwidth fractions $y$, so the
piecewise-polyhedral calmness of~\cite{wang2017fast} does not apply;
we use Robinson's metric regularity~\cite{robinson1976} with compactness of the iterate
set to obtain a uniform KKT error bound. Third, our multi-block
decomposition with consensus constraints separates blocks at the
problem level, removing the need for the inexact Uzawa linearization
that~\cite{wang2017fast} uses to make its two-block routing
subproblem separable; subproblems retain $C_{ml}$ and $g_{w,m}$
directly, and fully parallel (Jacobi-style) updates achieve global
linear convergence under the proximal regularization sized in
Assumption~\ref{assump:prox}. Together, these yield the first distributed parallel ADMM for
cross-layer covert network control, with Q-linear convergence under
aggregate energy-detection constraints.}
}
\section{System Model} \label{sec:formulation}

\subsection{Multi-Modal, Multi-Hop Network Model}
Consider a multi-hop wireless network represented by a directed graph $\mathcal{G} = (\mathcal{N}, \mathcal{L})$, where $\mathcal{N}$ is the set of nodes and $\mathcal L\subseteq\mathcal N\times\mathcal N$ is the set
of directed transmitter--receiver links between distinct nodes.

Let $\mathcal{F}$ denote the set of traffic flows in the network. Each flow $f\in\mathcal{F}$ is specified by a source node $s_f\in\mathcal{N}$, a destination node $d_f\in\mathcal{N}\setminus\{s_f\}$, and an admitted rate $x_f$ satisfying
\begin{equation}
    x_f^{\min} \leq x_f \leq x_f^{\max},
    \label{eq:flow_bounds}
\end{equation}
where $x_f^{\min}$ and $x_f^{\max}$ are the minimum and maximum admissible rates of flow $f$, respectively. Associated with each flow $f$ is a utility function $U_f(x_f)$, which quantifies the benefit obtained by admitting rate $x_f$ into the network. We assume that $U_f(\cdot)$ is nondecreasing, twice continuously differentiable, and strongly concave on $[x_f^{\min},x_f^{\max}]$. 


Each link may operate over one or more communication modalities, such as VHF, UHF, and cellular. 
Let $\mathcal{M}$ denote the set of available communication modalities. Each link $l\in\mathcal{L}$ is associated with a modality set $\mathcal{M}_l\subseteq\mathcal{M}$ containing the modalities supported by link $l$. For each modality $m\in\mathcal{M}$, let $\mathcal{L}_m \triangleq \{l\in\mathcal{L}: m\in\mathcal{M}_l\}$
denote the set of links that can operate on modality $m$.

Let $r_{ml}^d \geq 0$ denote the rate transmitted over link $l\in\mathcal{L}$ and modality $m\in\mathcal{M}_l$ for traffic destined for node $d\in\mathcal{N}$. 
The link rates satisfy the following flow-conservation constraints:
\begin{align}
    \!\!\!\!\sum_{l\in\mathcal{O}(n)} \sum_{m \in \mathcal{M}_l}\!\! r_{ml}^d
    -\!\!\!
    \sum_{l\in\mathcal{I}(n)} \sum_{m \in \mathcal{M}_l}\!\! r_{ml}^d
    &=\!\!
    \sum_{f\in\mathcal{F}}
    \mathbf{1}_{\{s_f=n,\, d_f=d\}} x_f, \!
    \label{eq:flow_conservation}
\end{align}
for all $n,d\in\mathcal{N}, n\neq d,$ where $\mathcal{O}(n)$ and $\mathcal{I}(n)$ denote the sets of outgoing and
incoming links of node $n$, respectively. 

For each link $l\in\mathcal{L}$ and modality $m\in\mathcal{M}_l$, let
$y_{ml}\in[0,1]$ denote the fraction of the bandwidth of modality $m$
allocated to link $l$, and let $p_{ml}\in [0,p^{\max}_{ml}]$ denote the
transmit power. The Shannon capacity of link $l$ under modality $m$ is
\begin{equation}
    C_{ml}(p_{ml},y_{ml})
    =
    y_{ml}\Omega_m 
    \log_2\!\left(
        1+
        \frac{p_{ml}|h_{ml}|^2}
        {y_{ml}\Omega_m N_{0,m}}
    \right),
    \label{eq:capacity}
\end{equation}
where $\Omega_m$ denotes the total bandwidth of modality $m$,
$y_{ml}\Omega_m$ is the bandwidth allocated to link $l$ under modality
$m$, $h_{ml}$ denotes the channel coefficient of link $l$ under modality
$m$, and $N_{0,m}$ denotes the receiver noise Power Spectral Density
(PSD). The function
$C_{ml}(p,y)$ is the perspective transform ~\cite[Section 3.2.6]{boyd2004convex}
of a logarithmic function and is therefore jointly concave in $(p,y)$.
We adopt the continuous extension $C_{ml}(p,0)=0$ when $y=0$.
For each link $l\in\mathcal L$ and modality
$m\in\mathcal M_l$, let
$q_{ml}=\sum_{d\in\mathcal N}r_{ml}^{d}$ denote the aggregate rate
carried by link $l$ over
modality $m$. The aggregate rate must not exceed the corresponding link
capacity:
\begin{equation}\label{eq:capacity_bound}
q_{ml} \leq C_{ml}(p_{ml},y_{ml}),
\qquad \forall l\in\mathcal L,\ m\in\mathcal M_l.
\end{equation}
We assume that the channel state $h_{ml}$ is perfectly known at the
transmitter of link $l$.


\subsection{Multi-Modal, Node-Exclusive Interference Model}

We adopt a bandwidth-fraction version of the node-exclusive
interference model~\cite{eryilmaz2006joint,lin2006utility}.
For each node $n$ and modality $m$, the total bandwidth fraction
allocated to its incoming and outgoing links cannot exceed one.
Thus, links incident to each node $n$ use nonoverlapping portions
of the bandwidth of modality $m$, although they may operate
simultaneously over different portions of that
bandwidth.

Recall that not every link supports every modality. Therefore, the node-exclusive interference model is characterized by the following constraints:
\begin{equation}
    \sum_{l\in (\mathcal{I}(n)\cup\mathcal{O}(n))\cap\mathcal{L}_m}
    y_{ml}
    \leq 1,~~~\forall\, n\in\mathcal{N},\; \forall\, m\in\mathcal{M}.
    \label{eq:degree}
\end{equation}
where, for each node $n$ and modality $m$, the corresponding interference constraint involves only the links in $\mathcal L_m$ that are outgoing and incoming links of node $n$.

\subsection{Covert Energy Detection Model}
 
Let $\mathcal{W}$ denote the set of adversarial Willies. For each modality $m$, let $\mathcal{W}_m \subseteq \mathcal{W}$ denote the subset of Willies monitoring modality $m$. If the legitimate network nodes do not know which Willies are monitoring modality $m$, we adopt the worst-case assumption that $\mathcal{W}_m = \mathcal{W}$, meaning that all Willies monitor each modality $m$. Since each Willie
$w \in \mathcal{W}_m$ has no prior knowledge of the bandwidth fractions $y_{ml}$, it cannot isolate individual links. Instead, Willie monitors the entire bandwidth $\Omega_m$ of modality $m$ and observes the superposition of all signals transmitted over that modality. Willie $w$ applies an energy detector to this aggregate received signal to distinguish between silence and transmission over modality $m$. 

For each modality $m\in\mathcal{M}$, define
$\mathbf{p}_m\triangleq(p_{ml})_{l\in\mathcal{L}_m}$ as the vector of
transmit powers for all links operating on modality $m$. The aggregate per-sample signal-to-noise ratio (SNR) observed by Willie $w$ on modality $m$ is then given by
\begin{equation}
s_{w,m}=\sum_{l\in\mathcal{L}_m}A_{w,ml}\,p_{ml},
\qquad
A_{w,ml}\triangleq\frac{|h_{w,ml}|^{2}}{\Omega_m N_{0,w,m}},
\label{eq:snr}
\end{equation}
where $h_{w,ml}$ is the channel coefficient from the transmitting node of link $l$ to Willie $w$, and
$N_{0,w,m}$ is the noise PSD at Willie $w$, both on modality $m$.
 
For energy detection on modality $m$, Willie $w$ observes the channel over a duration of $T_{w,m}$ seconds. The number of degrees of freedom, or equivalently, the number of complex baseband samples, is given by the effective time-bandwidth product:
\begin{equation}
    L_{w,m} \triangleq \zeta_{w,m} T_{w,m}\Omega_m,
    \label{eq:samples}
\end{equation}
where $\zeta_{w,m} \in (0,1]$ denotes the bandwidth utilization factor capturing potentially non-ideal sampling effects.

Consider the binary hypothesis testing problem at Willie $w$ between silence and transmission over modality $m$. Willie $w$ assumes equal prior probabilities $1/2$ for silence and transmission. Let $\theta>0$ denote the detection
threshold. Following the standard covert-communication convention, the Detection Error Probability (DEP) at threshold $\theta$ is defined as the sum of the false-alarm and missed-detection probabilities~\cite{yan2019,kong2022covert}:
\begin{equation}
    \xi(\theta) = P_{FA}(\theta) + P_{MD}(\theta).
\end{equation}
The minimum DEP under the optimal threshold is given by
\begin{equation}
    \mathrm{DEP}^*(s, L) \triangleq \min_{\theta > 0}\, \xi(\theta),
    \label{eq:dep_min}
\end{equation}
which admits the following closed-form expression~\cite{yan2019}:
\begin{equation}
\begin{aligned}
\!\!\!\!    \mathrm{DEP}^*(s, L)
    &= 1
    - \frac{\gamma\!\big(L,\, L(1 + \tfrac{1}{s})\ln(1+s)\big)}{\Gamma(L)}\\
    &\quad + \frac{\gamma\!\big(L,\, \tfrac{L}{s}\ln(1+s)\big)}{\Gamma(L)},
    \quad s > 0,\; L > 0,\!\!\!
\end{aligned}
\label{eq:dep_star}
\end{equation}
where $s$ denotes Willie’s aggregate received SNR, $L$ denotes the number of degrees of freedom for the energy detector, $\Gamma(a) = \int_0^\infty t^{a-1} e^{-t} dt$ for $a>0$ is the complete gamma function, and $\gamma(a, x) = \int_0^x t^{a-1} e^{-t} dt$ for $a>0$ and $x \ge 0$ is the lower incomplete gamma function.

The DEP at Willie $w$ on modality $m$ is a function of the
transmit-power vector $\mathbf{p}_m$, given by
\begin{equation}
\mathrm{DEP}_{w,m}(\mathbf{p}_m)
=
\mathrm{DEP}^{*}(s_{w,m},L_{w,m}),
\label{eq:dep_modality}
\end{equation}
where $s_{w,m}$ and $L_{w,m}$ are defined in~\eqref{eq:snr} and~\eqref{eq:samples}, respectively.
Under the optimal energy detector, $\mathrm{DEP}_{w,m}(\mathbf{p}_m)\in[0,1]$. 
A larger DEP indicates stronger covertness, while a smaller DEP indicates that Willie can detect the transmission more reliably.

\begin{figure}[!t]
\centering
\includegraphics[width=0.7\linewidth]{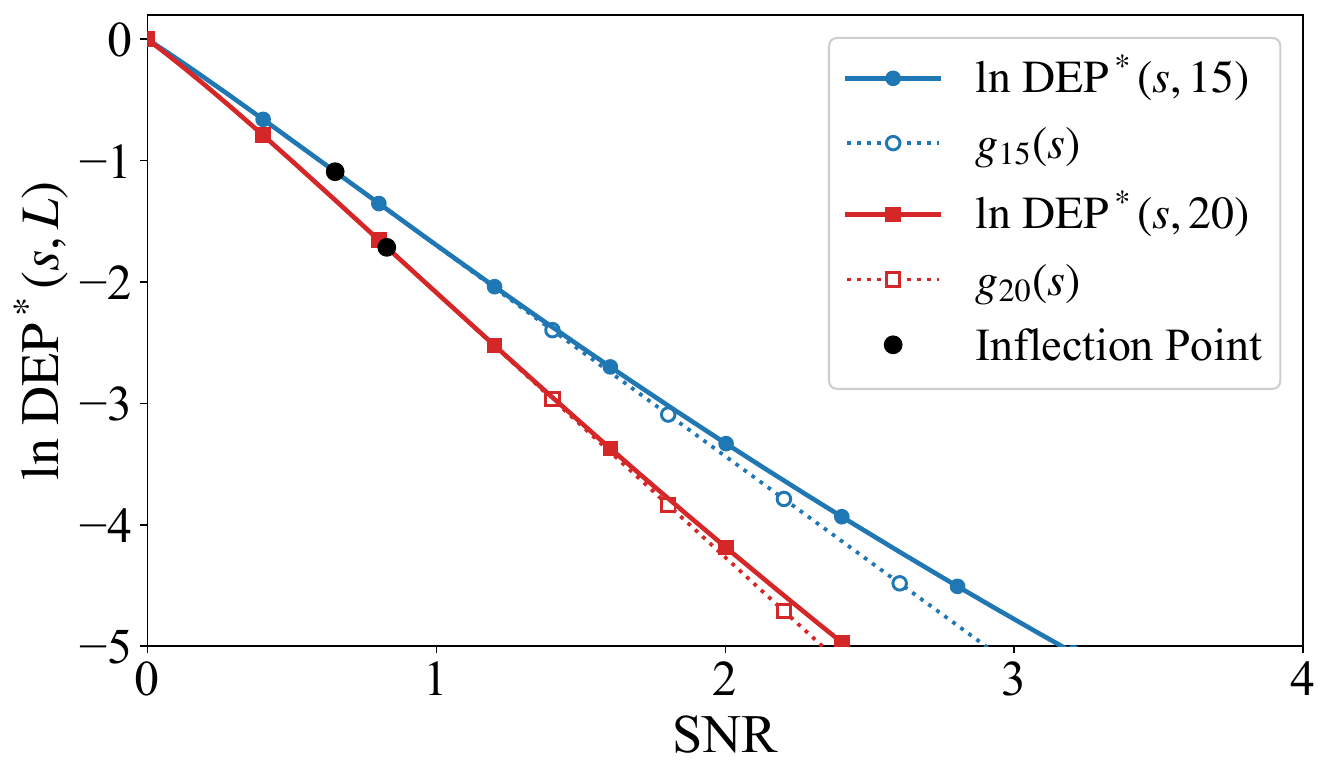}
\caption{The log-DEP function $s \mapsto\ln[\mathrm{DEP}^*(s,L)]$, which is concave on $[0,{s_L}]$ and
convex on $[{s_L},\infty)$, and its tightest concave lower bound $s \mapsto g_L({s})$, for
$L\in\{15,20\}$. }
\label{fig:lndep}
\vspace{-4mm}
\end{figure}

\section{Concave Covertness Bound and\\ Cross-Layer Covert Network Control Problems}

In this section, we construct a  concave lower bound on the log-DEP and formulate two cross-layer optimization problems: one with hard covertness constraints, and another that incorporates  both covertness utilities and hard covertness constraints.


\subsection{Tightest Concave Lower Bound on the log-DEP}
The function $s\mapsto\mathrm{DEP}^{*}(s, L)$ is non-concave. 
Consequently, the DEP 
function $\mathbf{p}_m \mapsto \mathrm{DEP}_{w,m}(\mathbf{p}_m)$ in \eqref{eq:dep_modality} is also non-concave. This non-concavity poses a key challenge in developing cross-layer covert network control algorithms. 

To enable a convex optimization formulation, we approximate the DEP function using a concave lower bound in the logarithmic domain. Specifically, we consider the log-DEP function $s \mapsto \ln [\mathrm{DEP}^{*}(s, L)]$, which is nearly concave (see Fig.~\ref{fig:lndep}) and exhibits the following curvature property. 
\begin{lemma} \label{lem:log_dep_curvature}
For each $L \geq 1$, the following assertions are true: 
\begin{itemize}
\item[(a)] $h_L(s) \triangleq \ln[\mathrm{DEP}^*(s, L)]$
is continuous on $[0,\infty)$ and twice continuously
differentiable on $(0,\infty)$;

\item[(b)] $h_L''(0^+)\triangleq
\lim_{s\downarrow0}h_L''(s)=K(1-K)$, where
$K \triangleq L^L e^{-L}/\Gamma(L)$;

\item[(c)] $h_L''(s) \sim L/s^2>0$ as $s \to \infty$.
Consequently, $h_L(s)$ is convex for all sufficiently large $s$.
\end{itemize}
\end{lemma}
\begin{IEEEproof}
\ifreport
See Appendix~\ref{apx:log_dep_curvature}.
\else
See \cite[App. A.]{tian2026covertfull}.
\fi
\end{IEEEproof}

Numerical results further suggest that there exists a point
\(s_L \geq 0\) such that \(h_L(s)\) is concave on
\([0,s_L]\) and convex on \([s_L,\infty)\).
When \(s_L>0\), it is the unique root of \(h_L''(s)=0\), which can be
computed efficiently by bisection.
When $L<6.45$, $K<1$ and $h_L''(0^+)>0$. In this case, $s_L=0$.
Conversely, when $L>6.45$, we have $K>1$ and $h_L''(0^+)<0$.
In this case, $s_L>0$.

We use these results to construct a concave lower bound \(g_L(s)\) of
\(h_L(s)\) by retaining its concave region and replacing its convex tail
with a tangent line, as shown in Fig.~\ref{fig:lndep}:
\begin{equation}
\!\!\! g_L(s)
=
\begin{cases}
h_L(s), 
&\!\!\! 0 \le s \le s_L,\\[0.5em]
h_L(s_L)
+
h_L'(s_L)(s-s_L),
&\!\!\! s > s_L.
\end{cases}\!\!\!
\label{eq:g_definition}
\end{equation}
Since $h_L(s)$ is convex for $s \ge s_L$, its tangent line at $s_L$
provides a lower bound on $[s_L,\infty)$. Therefore,
\begin{equation}
g_L(s)
\le
h_L(s)
=
\ln[\mathrm{DEP}^{*}(s,L)],
\quad s \ge 0.
\label{eq:g_star_lower_bound}
\end{equation}
Moreover, \(s \mapsto g_L(s)\) is concave because it
coincides with the concave function \(h_L(s)\) on \([0,s_L]\)
and extends linearly after \(s_L\) with the same slope. In fact,
\(g_L(s)\) is the tightest concave lower bound of the function $\ln[\mathrm{DEP}^*(s,L)]$. Define
\begin{equation}
g_{w,m}(\mathbf{p}_m)
\triangleq
g_{L_{w,m}}(s_{w,m}).
\label{eq:dep_bound_modality}
\end{equation}
Because $s_{w,m}$ is an affine function of $\mathbf p_m$ and
$g_L(s)$ is concave, the function
$\mathbf{p}_m \mapsto g_{w,m}(\mathbf{p}_m)$ is also concave.
Furthermore, by \eqref{eq:dep_modality},
\eqref{eq:g_star_lower_bound}, and
\eqref{eq:dep_bound_modality},
\begin{equation}
g_{w,m}(\mathbf p_m)
\le
\ln[\mathrm{DEP}_{w,m}(\mathbf p_m)].
\label{eq:g_lower_bound}
\end{equation}

\subsection{Problem 1: NUM with Hard Covertness Constraints}
The first NUM formulation imposes covertness through the hard constraints in \eqref{eq:p1-covert}. 
The resulting NUM problem is

\begin{subequations}\label{prob:P1}
\begin{align}
 &\max_{\{x_f,r_{ml}^{d},q_{ml},y_{ml},p_{ml},s_{w,m}\}}
    \sum_{f \in \mathcal{F}} U_f(x_f)
    \label{eq:P1_obj}\\
    \text{s.t.}~~
    & ~~\sum_{l\in\mathcal{O}(n)} \sum_{m \in \mathcal{M}_l} r_{ml}^d
    -
    \sum_{l\in\mathcal{I}(n)} \sum_{m \in \mathcal{M}_l} r_{ml}^d \nonumber\\
    &~~= \sum_{f \in \mathcal{F}}
    \mathbf{1}_{\{s_f=n,\,d_f=d\}} x_f,~~
    \forall\, n, d \in \mathcal{N},\, n \neq d,
    \label{eq:P1_flow}\\
    &~~q_{ml}=\sum_{d\in\mathcal N}r_{ml}^{d},\quad
\forall l\in\mathcal L,\ m\in\mathcal M_l,
\label{eq:P1_rate}\\
&~~q_{ml}
\leq C_{ml}(p_{ml},y_{ml}),
\quad
\forall l\in\mathcal L,\ m\in\mathcal M_l,
\label{eq:P1_cap}\\
    &~~ y_{ml}\geq0,\quad
\forall l\in\mathcal L,\ m\in\mathcal M_l,
\nonumber\\[-1mm]
&~~ \sum_{l\in(\mathcal I(n)\cup\mathcal O(n))\cap\mathcal L_m}
\!\!\!\!y_{ml}\leq1,\!\!\quad
\forall n\in\mathcal N,\ m\in\mathcal M,
\label{eq:P1_sched}\\
    &~~s_{w,m} = \sum_{l\in \mathcal{L}_m} A_{w,ml}\, p_{ml}, \forall m\in\mathcal{M},\ w\in\mathcal{W}_m, \label{eq:p1-snragg}\\
    &~~g_{L_{w,m}}(s_{w,m}) \ge \ln\eta_{w,m}, \forall m\in\mathcal{M},\ w\in\mathcal{W}_m, \label{eq:p1-covert}\\
    &~~ x_f^{\min} \leq x_f \leq x_f^{\max},
    \quad \forall\, f \in \mathcal{F},
    \label{eq:P1_xbnd}\\
    &~~ r^d_{ml} \geq 0,
    \quad \forall\, l \in \mathcal{L},\forall m\in\mathcal{M}_l, \forall\, d \in \mathcal{N},
    \label{eq:P1_rbnd}\\
    &~~ 0 \leq p_{ml} \leq p^{\max}_{ml},
    \quad \forall l\in\mathcal{L},\; \forall m\in\mathcal{M}_l.
    \label{eq:P1_pbnd}
\end{align}
\end{subequations}
The objective \eqref{eq:P1_obj} maximizes the total network utility. 
Constraints \eqref{eq:P1_flow}--\eqref{eq:P1_sched} enforce flow conservation, aggregate-rate constraint, link-capacity feasibility, and node-exclusive scheduling, respectively. Constraint \eqref{eq:p1-snragg} defines the aggregate
SNR observed by Willie $w$ on modality $m$, while~\eqref{eq:p1-covert} requires the corresponding DEP to satisfy the prescribed
covertness level $\eta_{w,m}$, where $\eta_{w,m}$ is the minimum acceptable DEP for
Willie $w$ and modality $m$. Since $g_{w,m}(\mathbf{p}_m)$ is
a lower bound on $\ln[\mathrm{DEP}_{w,m}(\mathbf{p}_m)]$, \eqref{eq:dep_bound_modality}, \eqref{eq:g_lower_bound},
and~\eqref{eq:p1-covert} together guarantee:
\begin{align}
    \mathrm{DEP}_{w,m}(\mathbf{p}_{m})
    \ge \eta_{w,m}.
\end{align}
\subsection{Problem 2: NUM with Covertness Utilities and Constraints}
The second formulation further incorporates covertness into the objective through a utility function. Let $G_w(\cdot)$ be a nondecreasing, twice continuously differentiable, and strongly concave function. The second problem is formulated as
\begin{subequations}\label{prob:unified}
\begin{align}
   &\max_{
\{x_f,r_{ml}^{d},q_{ml},y_{ml},p_{ml},s_{w,m}\}
} \sum_{f \in \mathcal{F}} U_f(x_f) \nonumber\\&\hspace{27mm}+\!\! \sum_{m \in \mathcal{M}} \sum_{w \in \mathcal{W}_m} \!\!G_w\big(g_{L_{w,m}}(s_{w,m})\big) \!\!\! \label{eq:P2_obj} \\
    &\qquad\qquad\ \text{s.t.} ~~~~~~ \hspace{2mm}\eqref{eq:P1_flow} - \eqref{eq:P1_pbnd}. \label{eq:P2_con}
\end{align}
\end{subequations}
\section{Distributed Cross-Layer Algorithm Design}
\label{sec:algorithm}

We develop a Parallel Proximal ADMM  algorithm to solve these two problems.

\subsection{Consensus Constraints for Distributed Scheduling}
To enable distributed scheduling, we introduce local copies
of the scheduling variables so that each node updates its own
scheduling variables locally and exchanges information only
with neighboring nodes. For each link $l=(i,j)\in\mathcal{L}$
and modality $m\in\mathcal{M}_l$, each endpoint
$n\in\{i,j\}$ maintains a local copy $y_{ml}^{(n)}$ of the link-level
scheduling variable $y_{ml}$. The two endpoint copies satisfy
the consensus constraint
\begin{equation}
y_{ml}^{(i)}=y_{ml}^{(j)},
\quad
\forall\,l=(i,j)\in\mathcal L,\quad m\in\mathcal M_l.
\label{eq:consensus}\tag{16k}
\end{equation}

The link-capacity constraint \eqref{eq:P1_cap} and the
node-exclusive constraint \eqref{eq:P1_sched} can be
equivalently expressed as \eqref{eq:consensus} and the
following two local constraints on each node's own copies:
\begin{align}
\!\!\!\!&0\leq q_{ml}
\!\leq \!C_{ml}\bigl(p_{ml},y_{ml}^{(i)}\bigr),
\ \forall l=(i,j)\in\mathcal{L},
m\in\mathcal{M}_l,
\label{eq:local_cap}\tag{16l}\\
& y_{ml}^{(n)}\geq0,\quad
\forall n\in\mathcal N,\ m\in\mathcal M,\ 
l\in(\mathcal I(n)\cup\mathcal O(n))\cap\mathcal L_m,
\nonumber\\[-1mm]
& \sum_{l\in(\mathcal I(n)\cup\mathcal O(n))\cap\mathcal L_m}
y_{ml}^{(n)}\leq1,\quad
\forall n\in\mathcal N,\ m\in\mathcal M,
\label{eq:local_exclusive}\tag{16m}
\end{align}
Here, \eqref{eq:local_cap} uses the transmitter $i$'s local
copy $y_{ml}^{(i)}$. Therefore, Problems~\eqref{prob:P1} and~\eqref{prob:unified}
refer to these equivalent endpoint-local forms.

\subsection{Augmented Lagrangian Formulation}
Define the dual variables $\lambda_n^d\in\mathbb R$ for
\eqref{eq:P1_flow}, $\mu_{ml}\in\mathbb R$ for
\eqref{eq:P1_rate}, $\chi_{w,m}\in\mathbb R$ for
\eqref{eq:p1-snragg}, and $\omega_{ml}\in\mathbb R$ for  \eqref{eq:consensus}.
Since each of these four constraints is a linear equality, all
multipliers are unconstrained.
Let $\boldsymbol\lambda$, $\boldsymbol\mu$,
$\boldsymbol\chi$, and $\boldsymbol\omega$ denote the vectors
that stack the dual variables $\lambda_n^d$, $\mu_{ml}$,
$\chi_{w,m}$, and $\omega_{ml}$ over their respective index
sets. Collect the primal variables $(x_f)_{f\in\mathcal F}\!$,
$(r_{ml}^{d})_{l\in\mathcal L,\,m\in\mathcal M_l,\,d\in\mathcal N}$,
$(p_{ml})_{l\in\mathcal L,\,m\in\mathcal M_l}$,
$(q_{ml})_{l\in\mathcal L,\,m\in\mathcal M_l}$,
$(y_{ml}^{(n)})_{l=(i,j)\in\mathcal L,\,m\in\mathcal M_l,\,
n\in\{i,j\}}$, and
$(s_{w,m})_{m\in\mathcal M,\,w\in\mathcal W_m}$ into the
block vectors $\mathbf x,\mathbf r,\mathbf p,\mathbf q,
\mathbf y^{(n)},$ and $\mathbf s$, respectively.


The augmented Lagrangian is given by
\begin{align}
&\mathcal L_\rho\!\big(
\mathbf x,\mathbf r,\mathbf p,\mathbf q,
\mathbf y^{(n)},\mathbf s,
\boldsymbol\lambda,\boldsymbol\mu,
\boldsymbol\chi,\boldsymbol\omega
\big) \nonumber\\
&= \sum_{f\in\mathcal{F}} U_f(x_f)
+ \sum_{m\in\mathcal{M}}\sum_{w\in\mathcal{W}_m}
G_w\big(g_{L_{w,m}}(s_{w,m})\big) \nonumber\\
& - \frac{\rho}{2}
\sum_{\substack{n,d\in\mathcal{N}\\ n\ne d}}
\Bigg[
\sum_{l\in\mathcal{O}(n)}\sum_{m\in\mathcal{M}_l} r_{ml}^d
- \sum_{l\in\mathcal{I}(n)}\sum_{m\in\mathcal{M}_l} r_{ml}^d
\nonumber\\
&\hspace{9.2em}
- \sum_{f\in\mathcal{F}}
\mathbf{1}_{\{s_f=n,\,d_f=d\}}x_f
+ \frac{\lambda_n^d}{\rho}
\Bigg]^2
\nonumber\\
& - \frac{\rho}{2}
\sum_{l\in\mathcal{L}}\sum_{m\in\mathcal{M}_l}
\Bigg[
\sum_{d\in\mathcal{N}} r_{ml}^d
- q_{ml}
+ \frac{\mu_{ml}}{\rho}
\Bigg]^2 \nonumber\\
& - \frac{\rho}{2}
\sum_{m\in\mathcal{M}}\sum_{w\in\mathcal{W}_m}
\Bigg[
s_{w,m}
- \sum_{l\in\mathcal{L}_m} A_{w,ml}p_{ml}
+ \frac{\chi_{w,m}}{\rho}
\Bigg]^2 \nonumber\\
&-\frac{\rho}{2}
\sum_{l=(i,j)\in\mathcal L}
\sum_{m\in\mathcal M_l}
\left[
y_{ml}^{(i)}-y_{ml}^{(j)}
+\frac{\omega_{ml}}{\rho}
\right]^2,
\label{eq:auglag}
\end{align}
where $\rho>0$ is the penalty parameter. For Problem~\eqref{prob:P1}, the term
$\sum_{m\in\mathcal{M}}\sum_{w\in\mathcal{W}_m}
G_w(g_{L_{w,m}}(s_{w,m}))$
is omitted. Hence, \eqref{eq:auglag} penalizes only the four linear equalities,
constraints~\eqref{eq:local_cap} and \eqref{eq:p1-covert} remain hard constraints.

\subsection{Parallel Proximal ADMM Algorithm}\label{sec:algo_design}
We now present the PP-ADMM algorithm, which is a distributed
cross-layer optimization algorithm for solving both problems.
At iteration $k+1$, the local primal variables associated with
(i) transport-layer congestion control,
(ii) network-layer routing,
(iii) link-layer power, rate, and scheduling control,
and (iv) aggregate SNR are updated in parallel by solving the corresponding
block subproblems induced by the augmented Lagrangian
$\mathcal{L}_{\rho}$ in \eqref{eq:auglag}, using iterate-$k$
values of the other primal variables. The dual variables are then updated by gradient ascent. To
stabilize the parallel local updates and ensure strong
concavity of each local primal subproblem, we include proximal
regularization terms
$\frac{\alpha}{2}(x_f-x_f^k)^2$,
$\frac{\alpha}{2}(r_{ml}^d-r_{ml}^{d,k})^2$,
$\frac{\alpha}{2}(p_{ml}-p_{ml}^k)^2$,
$\frac{\alpha}{2}(q_{ml}-q_{ml}^k)^2$,
$\frac{\alpha}{2}(y_{ml}^{(n)}-y_{ml}^{(n),k})^2$, and
$\frac{\alpha}{2}(s_{w,m}-s_{w,m}^k)^2$, where
$\alpha>0$ is the proximal weight.

In the proposed \textbf{Parallel Proximal ADMM
(PP-ADMM) Algorithm}, the updates at iteration $k+1$ are given below.



\subsubsection{Congestion Control}
The source node of flow $f$ solves
\vspace{-3mm}
\begin{align}
&x_f^{k+1} = \arg\max_{x_f^{\min} \leq x \leq x_f^{\max}} \bigg\{
U_f(x)
- \frac{\alpha}{2}\left(x - x_f^k\right)^2 \nonumber \\
&- \frac{\rho}{2}\bigg(
\sum_{l \in \mathcal{O}(s_f)}\sum_{m \in \mathcal{M}_l}\!\! r_{ml}^{d_f,k}\!
-\!\!\!\! \sum_{l \in \mathcal{I}(s_f)}\sum_{m \in \mathcal{M}_l} \!\!\!r_{ml}^{d_f,k}\!
- \! x \! + \! \frac{\lambda_{s_f}^{d_f,k}}{\rho}
\bigg)^2
\bigg\}.
\label{eq:step_x}
\end{align}
\ignore{
\begin{align}
   \!\!\! &x^{k+1}_f =  \amax_{x_f^{\min} \leq x \leq x_f^{\max}}
    \bigg[ U_f(x)
    - \lambda^{d_f,k}_{s_f}\, x - \frac{\alpha_f}{2}(x - x^k_f)^2  \notag \\
     & ~~~~~- \frac{\rho}{2}
      \Bigl(\sum_{l \in \mathcal{O}(s_f)} \!\!r^{d_f,k}_l- \!\!\sum_{l \in \mathcal{I}(s_f)} \!\!r^{d_f,k}_l\!
      - x\Bigr)^{\!2}\! \bigg].\!\!\!
    \label{eq:step_x}
\end{align}
}
\vspace{-6mm}
\subsubsection{Routing}
For each link
$l = (i,j) \in \mathcal{L}$, modality $m \!\in\! \mathcal{M}_l$, and destination
$d \in \mathcal{N}$, the routing update is 
given by \vspace{-3mm}\begin{align}
&r_{ml}^{d,k+1} \nonumber\\
&\!\!=\! \arg\min_{r_{ml}^{d} \ge 0}\! \bigg\{ \!\frac{\rho}{2}\mathbf{1}_{\{i\neq d\}}\bigg(\!r_{ml}^{d} + z_i^{d,k} - r_{ml}^{d,k}\!\bigg)^{2}
 \!\!+ \!\frac{\rho}{2}\mathbf{1}_{\{j\neq d\}}\bigg(\!r_{ml}^{d} \nonumber\\
&\quad- z_j^{d,k} - r_{ml}^{d,k}\!\bigg)^{2} \! + \frac{\rho}{2}\Big(r_{ml}^{d} + \!\!\sum_{d'\in\mathcal{N}\setminus\{d\}} r_{ml}^{d',k} - q_{ml}^{k} + \frac{\mu_{ml}^{k}}{\rho}\Big)^{2} \nonumber\\
&\quad + \frac{\alpha}{2}\Big(r_{ml}^{d} - r_{ml}^{d,k}\Big)^{2} \bigg\}, \label{eq:routing}
\end{align}
where $z_d^{d,k}\triangleq0$, and, for $n\neq d$, 
\begin{align}
z_n^{d,k} \triangleq{}& \sum_{l'\in\mathcal{O}(n)}\sum_{m'\in\mathcal{M}_{l'}} r_{m'l'}^{d,k}
 - \sum_{l'\in\mathcal{I}(n)}\sum_{m'\in\mathcal{M}_{l'}} r_{m'l'}^{d,k} \nonumber\\
&- \sum_{f\in\mathcal{F}} \mathbf{1}_{\{s_f=n,\, d_f=d\}}\, x_f^{k} + \frac{\lambda_n^{d,k}}{\rho}.\label{eq:fb-source-residual}
\end{align} 
\ignore{

in closed-form (see Appendix \ref{} for its derivation)
\begin{equation}
    r^{d,k+1}_l = \bigg[ r^{d,k}_l
    + \frac{\lambda^{d,k}_i - \lambda^{d,k}_j - \pi^k_l}
    {\rho\beta_l} \bigg]_+,
    \tag{A2}\label{eq:step_r}
\end{equation}
where 
\begin{equation}
    \pi^k_l \triangleq \bigg[\mu^k_l + \rho\bigg(
    \sum_{d \in \mathcal{N}} r^{d,k}_l
    - \sum_{m \in \mathcal{M}_l}
    C_{ml}(p^k_{ml}, y^{(i),k}_{ml})
    \bigg)\bigg]_+.
    \label{eq:projected_price}
\end{equation}
}
\ignore{
Following the inexact Uzawa approach~\cite{wang2017fast}, we replace these couplings by a linear price $\pi^k_l$ at iterate $k$ plus a diagonal proximal term with weight $\beta_l$, where
\begin{equation}
    \pi^k_l \triangleq \bigg[\mu^k_l + \rho\bigg(
      \sum_{d \in \mathcal{N}} r^{d,k}_l
      - \sum_{m \in \mathcal{M}_l} C_{ml}(p^k_{ml}, y^k_{ml})
    \bigg)\bigg]_+.
    \label{eq:projected_price}
\end{equation}
The price $\pi^k_l$ is identical across destinations on link $l$. For each link $l = (i,j) \in \mathcal{L}$ and destination $d \in \mathcal{N}$, the closed-form update is
\begin{equation}
    r^{d,k+1}_l =
    \bigg[ r^{d,k}_l
    + \frac{\lambda^{d,k}_i - \lambda^{d,k}_j - \pi^k_l}{\rho\beta_l}
    \bigg]_+,
    \tag{A2}\label{eq:step_r}
\end{equation}
where $\lambda^{d,k}_i - \lambda^{d,k}_j$ acts as a generalized backpressure~\cite{neely2010stochastic}.
}
\vspace{-5mm}
\subsubsection{Power Control, Aggregate Rate, and Scheduling}
The per-link, per-modality capacity constraint
\eqref{eq:local_cap} couples $p_{ml}$, $q_{ml}$, and the
transmitter's local scheduling copy $y_{ml}^{(i)}$, so these
variables are updated jointly at each node. For each node
$n\in\mathcal{N}$ and modality $m\in\mathcal{M}$, node $n$
jointly updates the transmit powers and aggregate rates of its
outgoing links and its local scheduling variables
$(y_{ml}^{(n)})_{l\in(\mathcal{I}(n)\cup\mathcal{O}(n))
\cap\mathcal{L}_m}$
by solving
\begin{align}
&\Big(
\big(y_{ml}^{(n),k+1}\big)_{
l\in(\mathcal I(n)\cup\mathcal O(n))\cap\mathcal L_m},\
\big(p_{ml}^{k+1},q_{ml}^{k+1}\big)_{
l\in\mathcal O(n)\cap\mathcal L_m}
\Big)\nonumber\\
&=\arg\min_{\substack{
(y_{ml}^{(n)})_{
l\in(\mathcal I(n)\cup\mathcal O(n))\cap\mathcal L_m},\\
(p_{ml},q_{ml})_{l\in\mathcal O(n)\cap\mathcal L_m}}}
\Bigg\{
\frac{\rho}{2}
\sum_{l\in\mathcal{O}(n)\cap\mathcal{L}_m}
\Bigg[
\sum_{d\in\mathcal{N}}r_{ml}^{d,k}
\nonumber\\&\quad-q_{ml}+\frac{\mu_{ml}^{k}}{\rho}
\Bigg]^2+
\frac{\rho}{2}
\sum_{w\in\mathcal{W}_m}
\Bigg[
s_{w,m}^{k}
-\sum_{l\in\mathcal{O}(n)\cap\mathcal{L}_m}
A_{w,ml}p_{ml}
\nonumber\\
&\quad
-\sum_{l'\in\mathcal{L}_m\setminus\mathcal{O}(n)}
A_{w,ml'}p_{ml'}^{k}
+\frac{\chi_{w,m}^{k}}{\rho}
\Bigg]^2
\nonumber\\
&\quad+
\frac{\rho}{2}
\sum_{l=(n,j)\in\mathcal O(n)\cap\mathcal L_m}
\left[
y_{ml}^{(n)}-y_{ml}^{(j),k}
+\frac{\omega_{ml}^{k}}{\rho}
\right]^2
\nonumber\\[-1mm]
&\quad+
\frac{\rho}{2}
\sum_{l=(i,n)\in\mathcal I(n)\cap\mathcal L_m}
\left[
y_{ml}^{(n)}-y_{ml}^{(i),k}
-\frac{\omega_{ml}^{k}}{\rho}
\right]^2
\nonumber\\
&\quad+
\frac{\alpha}{2}
\sum_{l\in\mathcal{O}(n)\cap\mathcal{L}_m}
\left[
\big(p_{ml}-p_{ml}^{k}\big)^2
+\big(q_{ml}-q_{ml}^{k}\big)^2
\right]
\nonumber\\
&\quad+
\frac{\alpha}{2}
\sum_{l\in(\mathcal{I}(n)\cup\mathcal{O}(n))
\cap\mathcal{L}_m}
\big(y_{ml}^{(n)}-y_{ml}^{(n),k}\big)^2
\Bigg\}
\nonumber\\
&\text{s.t.}\quad
y_{ml}^{(n)}\geq0,\quad
\forall l\in(\mathcal I(n)\cup\mathcal O(n))\cap\mathcal L_m,
\nonumber\\
&\hspace{2.8em}
\sum_{l\in(\mathcal I(n)\cup\mathcal O(n))\cap\mathcal L_m}
y_{ml}^{(n)}\leq1,
\nonumber\\
&\hspace{2.8em}
0\leq p_{ml}\leq p_{ml}^{\max},\quad
0\leq q_{ml}\leq
C_{ml}\big(p_{ml},y_{ml}^{(n)}\big),
\nonumber\\
&\hspace{2.8em}
\forall l\in\mathcal O(n)\cap\mathcal L_m.
\label{eq:nodeupdate}
\end{align}

\subsubsection{Aggregate-SNR Update}
For each modality $m\in\mathcal{M}$ and each Willie $w\in\mathcal{W}_m$,
the auxiliary variable is updated in parallel by solving the scalar
subproblem
\vspace{-2mm}
\begin{align}
s_{w,m}^{k+1}&= \arg\max_{\substack{s\geq0,\\g_{L_{w,m}}(s)\geq\ln\eta_{w,m}}}\nonumber
\bigg\{\ G_w\big(g_{L_{w,m}}(s)\big)
\\&- \frac{\rho}{2}\left[s - \sum_{l\in\mathcal{L}_m} A_{w,ml}\, p_{ml}^{k} + \frac{\chi_{w,m}^{k}}{\rho}\right]^{2} \nonumber\\
&- \frac{\alpha}{2}\big(s - s_{w,m}^{k}\big)^{2} \bigg\}. \label{eq:s_update}
\end{align}
For Problem \eqref{prob:P1}, the covertness-utility
term $G_w(g_{L_{w,m}}(s))$ is omitted.

\subsubsection{Feasible Transmission Schedule}
For each link--modality pair $(m,l)$ with
$l=(i,j)\in\mathcal{L}_m$, the two endpoint nodes update their
local scheduling copies in \eqref{eq:nodeupdate}. During the iterations, a common node-exclusive schedule can
be formed as



\begin{align}
\!\!\!\!\!\bar{y}_{ml}^{k+1}
&\!=\!
\min\!\left\{
y_{ml}^{(i),k+1},
y_{ml}^{(j),k+1}
\right\},\!
\forall l=\!(i,j)\!\in\!\mathcal{L},
m\!\in\!\mathcal{M}_l.
\label{eq:feasible_schedule}
\end{align}
Since $\bar{y}_{ml}^{k+1}\leq y_{ml}^{(n),k+1}$ for
$n\in\{i,j\}$, the schedule in \eqref{eq:feasible_schedule}
satisfies the node-exclusive constraints. The two local copies agree at convergence.

\subsubsection{Dual Updates}
The dual variables are updated by gradient ascent on
(\ref{eq:auglag})
\begin{align}
&\lambda_n^{d,k+1} \!=\! \lambda_n^{d,k} + \tau\rho\Big(\sum_{l\in\mathcal{O}(n)}\sum_{m\in\mathcal{M}_l} \!r_{ml}^{d,k+1} \!-\! \sum_{l\in\mathcal{I}(n)}\sum_{m\in\mathcal{M}_l} \!r_{ml}^{d,k+1} \nonumber\\
& \qquad- \sum_{f\in\mathcal{F}}\mathbf{1}_{\{s_f=n,\, d_f=d\}}\, x_f^{k+1}\Big),\ \ \!\!\!\!\forall n,d\in\mathcal{N},\ n\ne d, \label{eq:lamupd}\\
&\mu_{ml}^{k+1} = \mu_{ml}^{k} + \tau\rho\Big(\sum_{d\in\mathcal{N}} r_{ml}^{d,k+1} - q_{ml}^{k+1}\Big),\nonumber \\ &\hspace{11.5em}\forall l\in\mathcal{L},\ \forall m\in\mathcal{M}_l, \label{eq:muupd}\\
&\chi_{w,m}^{k+1} = \chi_{w,m}^{k} + \tau\rho\Big(s_{w,m}^{k+1} - \sum_{l\in\mathcal{L}_m} A_{w,ml}\, p_{ml}^{k+1}\Big),\ \nonumber\\&\hspace{11.5em} \forall m\in\mathcal{M},\ w\in\mathcal{W}_m, \label{eq:chiupd}\\
&\omega_{ml}^{k+1} = \omega_{ml}^{k} + \tau\rho\Big(y_{ml}^{(i),k+1} - y_{ml}^{(j),k+1}\Big),\nonumber \\&\hspace{11.5em} \!\forall l=(i,j)\in\!\mathcal{L},\! \forall m\in\mathcal{M}_l, \label{eq:omegaupd}
\end{align}
where $\tau > 0$ is the dual step size.

\begin{figure*}[t]
\centering
\begin{minipage}[t]{0.23\textwidth}
    \centering
    \includegraphics[width=\linewidth]{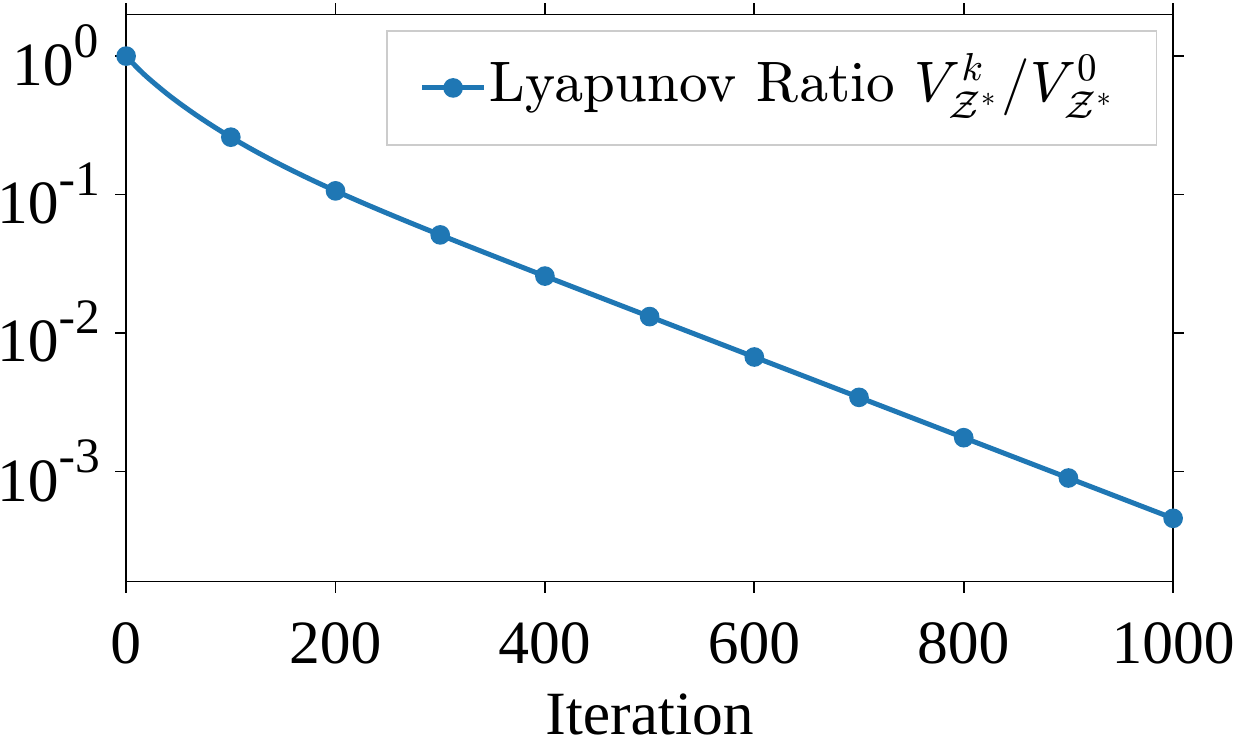}
    \captionof{figure}{Lyapunov ratio $V_{\mathcal{Z}^*}^k/V_{\mathcal{Z}^*}^0$ vs. $k$ under static channel.}
    \label{fig:lyapunov_fast}
\end{minipage}
\hfill
\begin{minipage}[t]{0.24\textwidth}
    \centering
    \includegraphics[width=\linewidth]{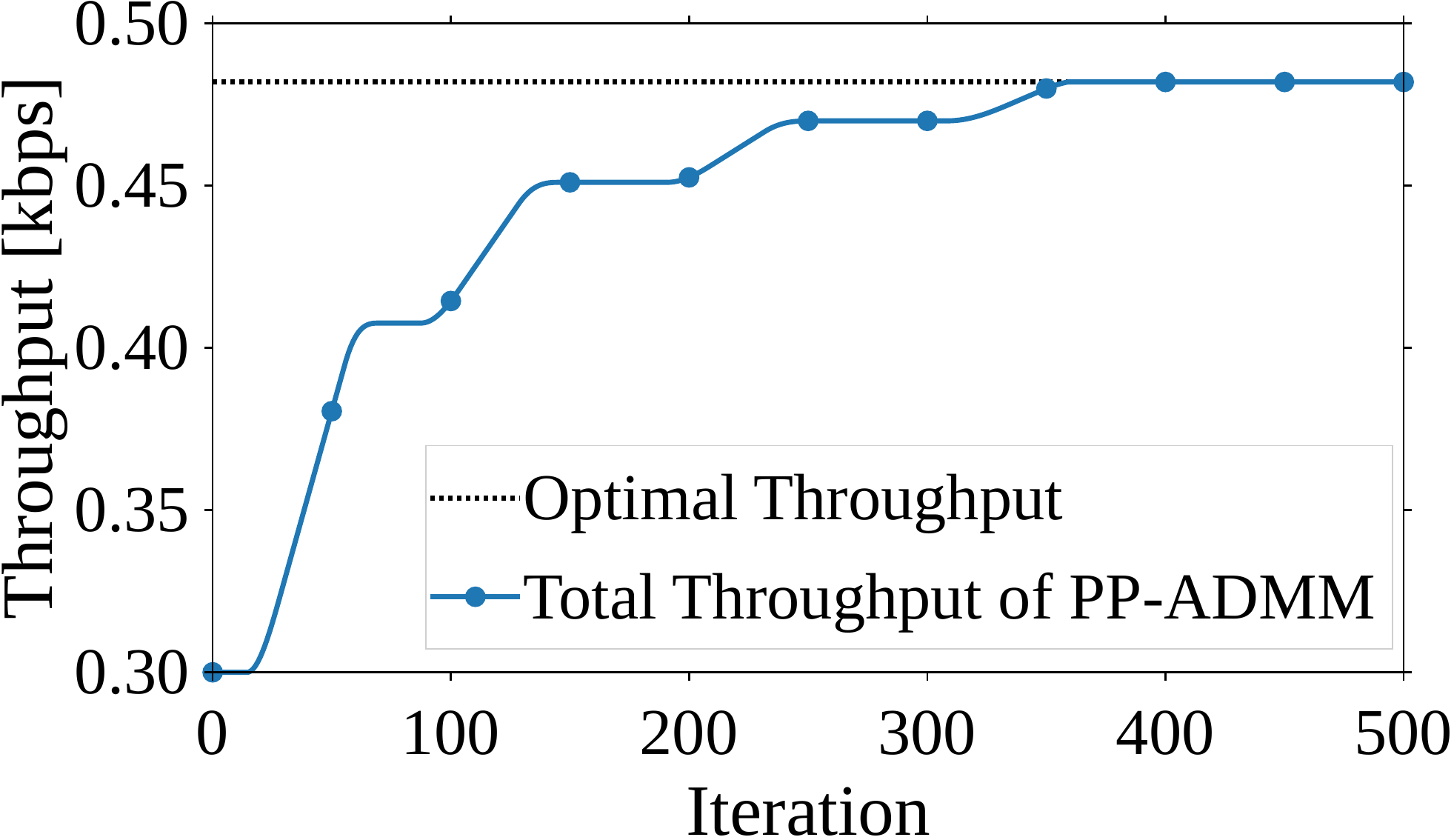}
    \captionof{figure}{Total throughput of two flows under static channel, converging to the  optimum.}
    \label{fig:utility_convergence}
\end{minipage}
\hfill
\begin{minipage}[t]{0.24\textwidth}
    \centering
    \includegraphics[width=0.98\linewidth]{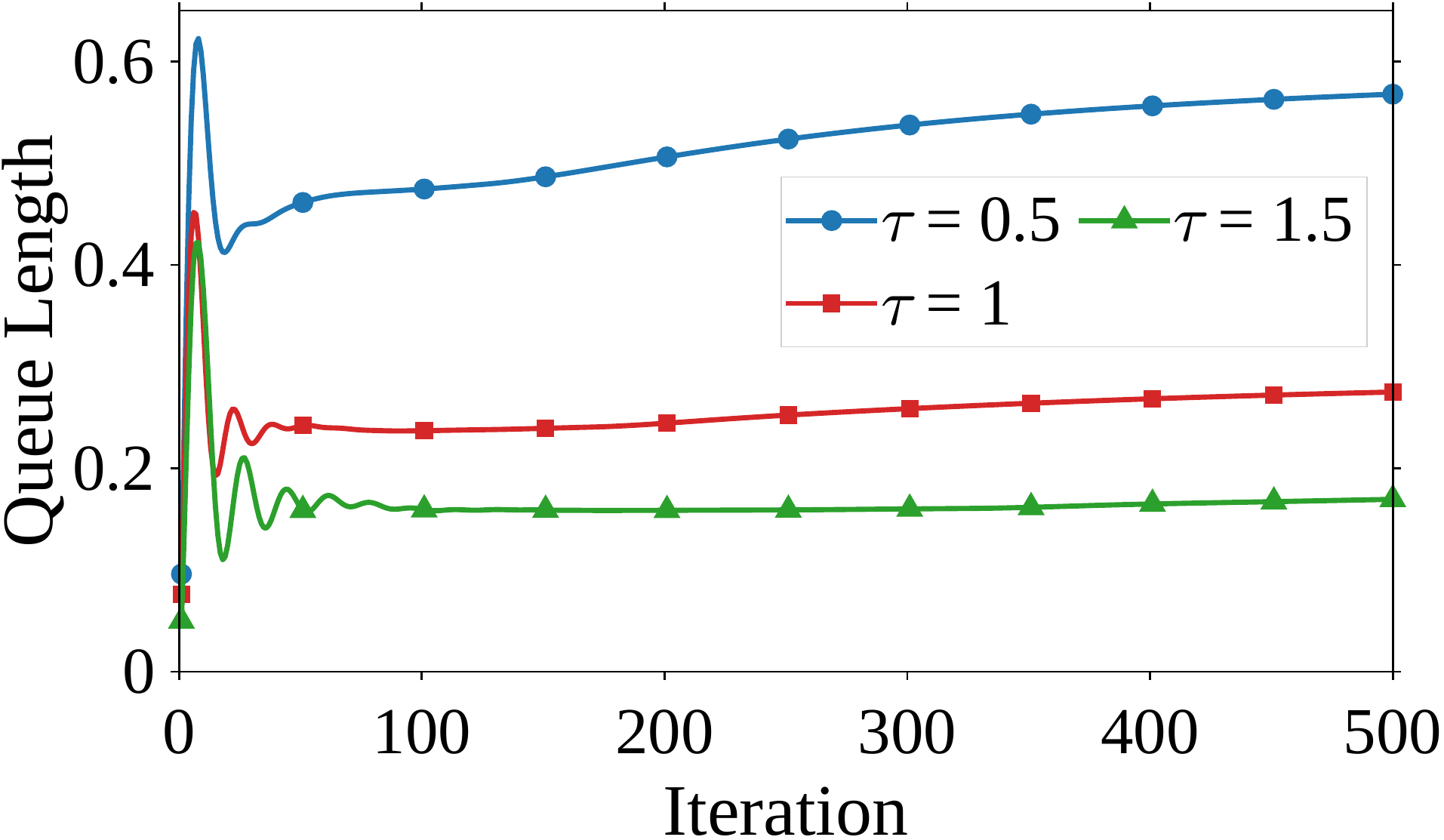}
    \captionof{figure}{Queue length under different dual step
sizes $\tau$ for a 5-node random
    network under static channels.}
    \label{fig:queue_length}
\end{minipage}
\hfill
\begin{minipage}[t]{0.24\textwidth}
    \centering
    \includegraphics[width=\linewidth]{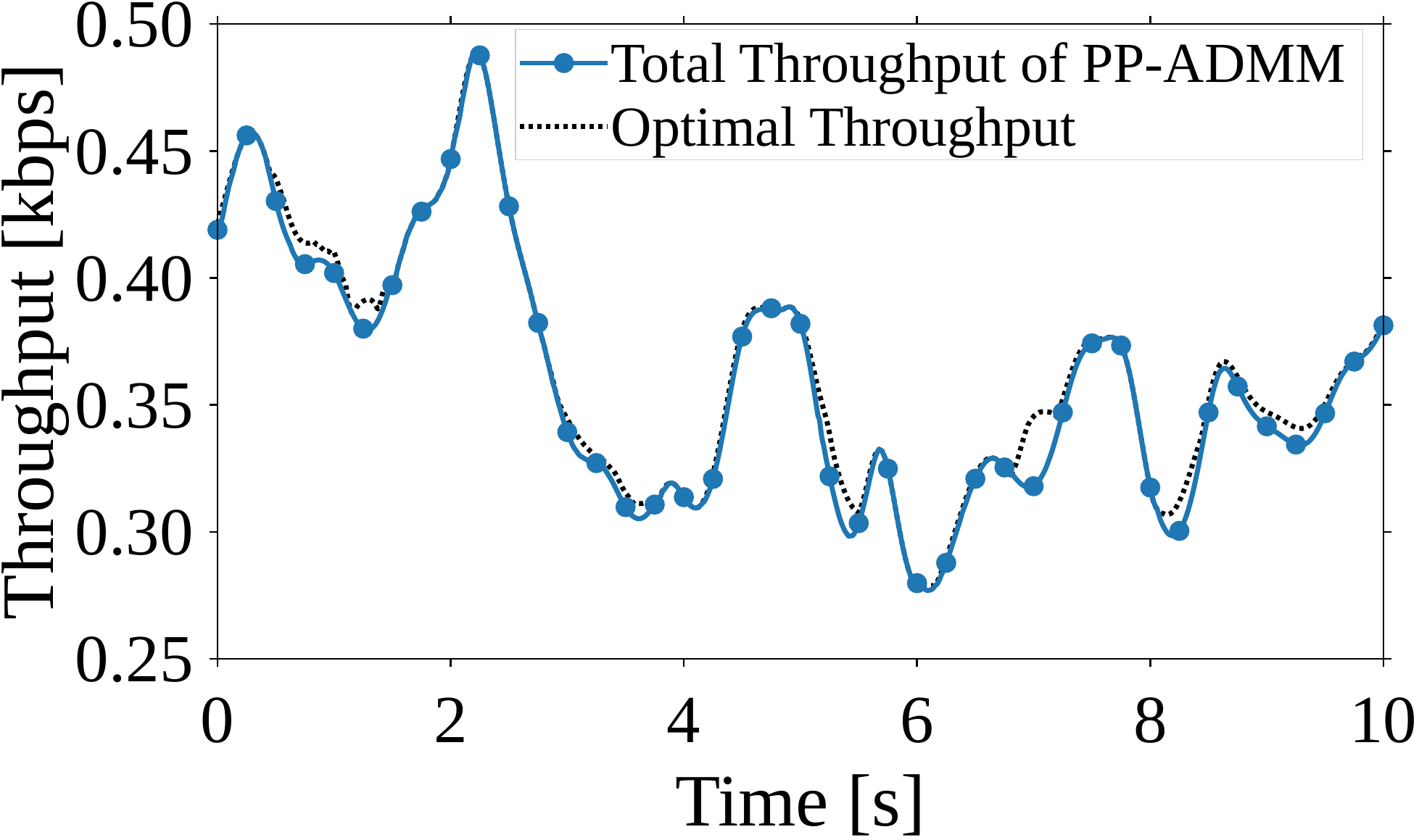}
    \captionof{figure}{Tracking of total throughput under Willie mobility and channel fading.}
    \label{fig:throughput_tracking}
\end{minipage}
\vspace{-5mm}
\end{figure*}

\section{Convergence Analysis}
\label{sec:convergence}
This section establishes global linear convergence of the PP-ADMM algorithm to the optimal solution set of Problem~\eqref{prob:unified}.
Let $\mathbf z\triangleq
(\mathbf x,\mathbf r,\mathbf p,\mathbf q,
\mathbf y^{(n)},\mathbf s,\boldsymbol\lambda,
\boldsymbol\mu,\boldsymbol\chi,\boldsymbol\omega)$
denote a primal--dual pair, and let $\mathbf z^k$ denote its value
after iteration $k$. Let $\|\cdot\|$ denote the Euclidean norm.

\begin{assumption}[Slater's Condition]
\label{assump:slater}
A strictly feasible primal solution exists for Problem~\eqref{prob:unified}.
\end{assumption}
The primal updates \eqref{eq:step_x}, \eqref{eq:routing}, \eqref{eq:nodeupdate}, and \eqref{eq:s_update}, together
with the dual updates \eqref{eq:lamupd}--\eqref{eq:omegaupd}, are equivalent to
standard parallel proximal ADMM with four primal blocks $\mathbf x$, $\mathbf r$,
$(\mathbf p,\mathbf q,\mathbf y^{(n)})$, and $\mathbf s$;
see
\ifreport
Appendix~B.
\else
\cite[App.~B]{tian2026covertfull}.
\fi
\begin{assumption}[Dual Step Size and Proximal Weight]
\label{assump:proximal}
Let $\rho>0$ and $0<\tau<2$. For some
$\delta\in(0,2-\tau)$, the proximal weight $\alpha>0$
satisfies
\begin{equation}
\alpha>
\rho\max_{1\le i\le4}
\lambda_{\max}\!\left(
\frac{4}{2-\tau-\delta}Q_i^{\mathsf T}Q_i-D_i
\right).
\label{eq:fb-alpha}
\end{equation}
Here, $Q_i$ is the coefficient matrix for block $i$ in
constraints \eqref{eq:P1_flow}, \eqref{eq:P1_rate}, \eqref{eq:p1-snragg}, and \eqref{eq:consensus}, $D_i$ contains the quadratic coefficients
from \eqref{eq:auglag} in the
corresponding primal updates
\eqref{eq:step_x}, \eqref{eq:routing}, \eqref{eq:nodeupdate}, and \eqref{eq:s_update} with the common factor $\rho/2$ removed,
and $\lambda_{\max}(M)$ denotes the largest eigenvalue
of a symmetric matrix $M$.
Explicit expressions for $Q_i$ and $D_i$ are given in
\ifreport
Appendix~\ref{app:fb-formulation}.
\else
\cite[App.~B]{tian2026covertfull}.
\fi
\end{assumption}

Let $\mathcal Z^*$ denote the set of optimal primal--dual pairs
of the problem under consideration. Let $\mathbf R(\mathbf z)$ denote the complete KKT residual, which measures
violation of the KKT optimality conditions. \ifreport
Detailed definitions are given in Appendix~\ref{apx:linear_convergence}.
\else
Due to space limitations, detailed definitions are given in
\cite[App.~C]{tian2026covertfull}.
\fi
It satisfies $\mathbf R(\mathbf z)=\mathbf 0$ if and only if
$\mathbf z\in\mathcal Z^*$.

\begin{assumption}[Local KKT Error Bound]
\label{assump:kkt_error}
For every $\mathbf z^*\in\mathcal Z^*$, there
exist constants $\epsilon_{\mathbf z^*}>0$ and
$\kappa_{\mathbf z^*}>0$ such that
\begin{equation}
\inf_{\widehat{\mathbf z}\in\mathcal Z^*}
\|\mathbf z-\widehat{\mathbf z}\|
\leq
\kappa_{\mathbf z^*}\|\mathbf R(\mathbf z)\|
\label{eq:kkt_error_bound}
\end{equation}
for every $\mathbf z$ with feasible local primal blocks
satisfying
$\|\mathbf z-\mathbf z^*\|<\epsilon_{\mathbf z^*}$.
\end{assumption}


For any $\mathbf z^*\in\mathcal Z^*$, define the
Lyapunov function
\begin{equation}
V(\mathbf z;\mathbf z^*)
\triangleq
\|\mathbf z-\mathbf z^*\|_G^2,
\label{eq:lyapunov}
\end{equation}
where $\|\mathbf v\|_G^2=\mathbf v^{\mathsf T}G\mathbf v$
and the fixed positive definite matrix $G$ is defined in
\ifreport
Appendix~\ref{apx:linear_convergence}.
\else
\cite[App.~C]{tian2026covertfull}.
\fi

\begin{theorem}[Global Linear Convergence]\label{thm:linear_convergence}
Under Assumptions \ref{assump:slater}--\ref{assump:kkt_error}, PP-ADMM converges globally and linearly to the
optimal primal--dual solution set of \eqref{prob:unified}. Specifically, for every
initial point $\mathbf{z}^0$, there exist
$\mathbf z^\infty\in\mathcal Z^*$, $\beta \in (0,1)$, and $C > 0$
such that the sequence $\{\mathbf{z}^k\}_{k\ge 0}$ generated by
PP-ADMM satisfies
\vspace{-2mm}\begin{align}
V_{\mathcal{Z}^*}^{k+1} &\le \beta\, V_{\mathcal{Z}^*}^{k}, \qquad k \ge 0, \label{eq:qlinear}\\
\big\|\mathbf{z}^k - \mathbf{z}^\infty\big\|^{2} &\le C \beta^k\, V_{\mathcal{Z}^*}^{0}, \qquad k \ge 0, \label{eq:rlinear}
\end{align}
where $V_{\mathcal Z^*}^k
\triangleq
\inf_{\mathbf z^*\in\mathcal Z^*}
V(\mathbf z^k;\mathbf z^*)$.
\end{theorem}

\begin{IEEEproof}
\ifreport See Appendix~\ref{apx:linear_convergence}\else See \cite[App. C]{tian2026covertfull}\fi.
\end{IEEEproof}
Theorem~\ref{thm:linear_convergence} also applies to Problem~\eqref{prob:P1} after removing the {$G_w(g_{L_{w,m}}(s))$} objective terms. 
The technical
novelties of Theorem~\ref{thm:linear_convergence} are summarized at the end of
Section~\ref{sec:intro}.
\section{Numerical Analysis}

We consider a covert multi-hop, multi-modal
network with $|\mathcal N|=5$, $|\mathcal L|=8$,
$|\mathcal M|=2$, $|\mathcal F|=2$, and $|\mathcal W|=2$,
operating over VHF and UHF terrestrial modalities. Each link supports
all modalities, $\mathcal M_l=\mathcal M$, and both Willies monitor
all modalities, $\mathcal W_m=\mathcal W$. The throughput utility function is
$U_f(x_f)=\ln x_f$ and the covertness utility function is
$G_w(s)=-e^{-s}$. All results are obtained by solving~\eqref{prob:unified}.

Channel power gains are modeled via distance-dependent path loss with
exponents $2.5$ (VHF) and $2.8$ (UHF), rich-scattering small-scale
fading, and omnidirectional antennas at all nodes and Willies. The
carrier frequencies are $f_{\mathrm{VHF}}^c=150$ MHz and
$f_{\mathrm{UHF}}^c=450$ MHz. The effective VHF and UHF bandwidths
are $\Omega_{\mathrm{VHF}}=5$ kHz and
$\Omega_{\mathrm{UHF}}=2.5$ kHz, with normalized receiver noise PSDs
$N_{0,\mathrm{VHF}}=0.5$, $N_{0,\mathrm{UHF}}=1.5$, and
$N_{0,w,m}=N_{0,m}$. Each Willie observes each monitored modality
for $T_{w,m}=10$ ms with $\zeta_{w,m}=1$, yielding degrees of freedom
$L_{w,\mathrm{VHF}}=50$ and $L_{w,\mathrm{UHF}}=25$. The default
parameters are $\eta_{w,m}=0.7$, $\rho=4$, and $\tau=1$. The proximal
weight $\alpha$ satisfies Assumption~\ref{assump:proximal}.

\begin{figure}[!t]
    \centering
    \includegraphics[width=0.7\columnwidth]{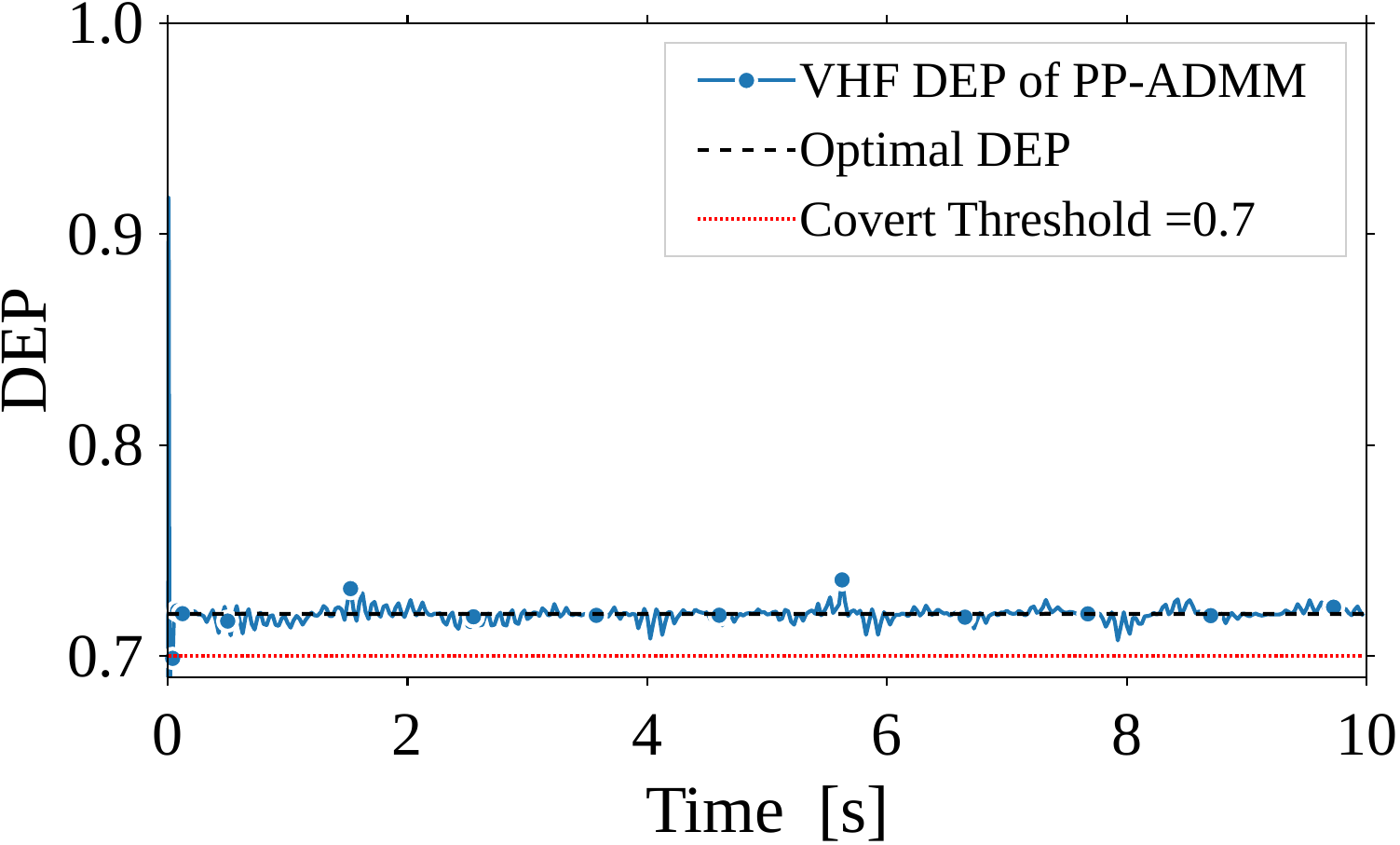}
    \caption{Tracking of VHF DEP of Willie 1 under Willie mobility and channel fading.}
    \label{fig:dep_tracking}
    \vspace{-3mm}
\end{figure}

\subsection{Linear Convergence and Scalability under Static Channels}

All distances and fading gains are fixed in the static-channel
experiments. Fig.~\ref{fig:lyapunov_fast} plots the Lyapunov ratio
$V_{\mathcal Z^*}^k/V_{\mathcal Z^*}^0$, confirming the exponentially fast 
convergence predicted by Theorem~\ref{thm:linear_convergence}.
Fig.~\ref{fig:utility_convergence} shows that the total throughput of the
two flows converges to the optimum.
Existing covert-routing works do not address the same joint
cross-layer problem and hence are not directly comparable.

Table~\ref{tab:network_size} reports the average iterations to
convergence and queue length over 100 random topologies per network
size $(|\mathcal N|,|\mathcal L|,|\mathcal F|)$, with $|\mathcal M|=2$ and $|\mathcal W|=2$. In each topology, the nodes and Willies are
placed uniformly at random, with constant node density, and the
closest node pairs are connected as bidirectional links, and channel gains follow the
distance-dependent path loss with exponents 2.5 (VHF) and 2.8 (UHF). Queue length
is the per-link magnitude of the flow-conservation multipliers
$\lambda_n^d$. Fig.~\ref{fig:queue_length} shows the queue evolution
for the 5-node network under different dual step sizes $\tau$.

\vspace{-2mm}
\subsection{Robust Tracking under Willie Mobility and Channel Fading}
\ignore{
Figs.~\ref{fig:throughput_vhf} and~\ref{fig:dep_convergence} illustrate the tracking performance of the aggregate throughput and VHF DEP, respectively, under Willie’s mobility and channel fading. The optimum for both metrics change over time. The proposed algorithm closely tracks the optimum over time for both metrics.
}
Figs.~\ref{fig:throughput_tracking}--\ref{fig:dep_tracking} show the
tracking performance of the total throughput and the DEP of Willie~1
for the VHF modality, respectively, for the 5-node network under
Willie mobility and channel fading. The legitimate nodes remain fixed,
while the Willies move at $v_W=5$ m/s. Under the Jakes fading model,
the maximum Doppler frequencies are
$f_{\mathrm{VHF}}^D=2.5$ Hz and $f_{\mathrm{UHF}}^D=7.5$ Hz.
Each algorithm iteration corresponds to a time slot of $1$ ms.
A video demonstration is provided
in~\cite{chakraborty2026video}. As the Willies move and channel conditions vary, the optimal values of both metrics evolve over time. The proposed algorithm closely tracks these time-varying optima, demonstrating effective tracking performance.

\section{Conclusion}
\label{sec:conclusion}

This paper developed a  Parallel Proximal ADMM algorithm for distributed multi-hop, multi-modal covert network control, achieving linear convergence and robust tracking under Willie mobility. Our future work will consider multi-antenna transmissions, Willie cooperation, and higher mobility speeds.

\begin{table}[!t]
\centering
\vspace{2mm}
\caption{PP-ADMM convergence speed and queue length.}
\label{tab:network_size}
{\footnotesize
\setlength{\tabcolsep}{9pt}
\renewcommand{\arraystretch}{1.2}
\begin{tabular}{|c|c|c|}
\hline
$(|\mathcal N|,|\mathcal L|,|\mathcal F|)$ & Iterations & Queue length \\
\hline
$(5,8,2)$  & 1204 & 1.166 \\
$(10,18,4)$ & 1936 & 1.383 \\
$(20,38,8)$ & 4339 & 1.972 \\
\hline
\end{tabular}}
\vspace{-2.5mm}
\end{table}
\section*{Acknowledgment}
We thank Hao Chen and Ananthram Swami for helpful discussions on this study. 



\bibliographystyle{IEEEtran}
\bibliography{refs}

\ifreport
\appendices
\input{Appendix}
\else
\fi

\end{document}

%% file: Appendix.tex
\section{Proof of Lemma~\ref{lem:log_dep_curvature}}
\label{apx:log_dep_curvature}

\subsection{Proof of part (a):}
For notational convenience, write
\begin{equation}
D(s) \triangleq \operatorname{DEP}^{*}(s,L).
\end{equation}
By the definition of $h_L(s)$ in
Lemma~\ref{lem:log_dep_curvature}, $h_L(s)=\ln D(s)$. For $s>0$, define
\begin{equation}\label{eq:lem1_a}
a(s) \triangleq
L\left(1+\frac{1}{s}\right)\ln(1+s),
\end{equation}
and
\begin{equation}\label{eq:lem1_b}
b(s) \triangleq \frac{L}{s}\ln(1+s).
\end{equation}

\noindent Because the upper incomplete gamma function can be written as
\(
\Gamma(L,\cdot)
\triangleq
\Gamma(L)-\gamma(L,\cdot),
\)
equation~\eqref{eq:dep_star} is equivalently rewritten as
\begin{equation}
D(s)
=
\frac{\Gamma\!\left(L,a(s)\right)
+\gamma\!\left(L,b(s)\right)}
{\Gamma(L)}.
\label{eq:dep_upper_gamma}
\end{equation}
For $s>0$, the functions $a(s)$ and $b(s)$ are infinitely differentiable.
For fixed $L\geq1$, the functions $x\mapsto\Gamma(L,x)$ and
$x\mapsto\gamma(L,x)$ are infinitely differentiable for $x>0$. Moreover,
$\Gamma\bigl(L,a(s)\bigr)>0$ and $\gamma\bigl(L,b(s)\bigr)\geq0$. Hence,
$D(s)>0$, and $h_L(s)=\ln D(s)$ is twice continuously differentiable on
$(0,\infty)$.

Furthermore, $a(s)\to L$ and $b(s)\to L$ as $s\downarrow0$. Therefore, by
continuity of the incomplete gamma functions,
\begin{equation*}
\lim_{s\downarrow0}D(s)
= \frac{\Gamma(L,L)+\gamma(L,L)}{\Gamma(L)}
= 1 .
\end{equation*}
At $s=0$, Willie's aggregate received SNR is zero, so the observation
distributions under silence and transmission are identical. Hence,
$P_{\mathrm{FA}}(\theta)+P_{\mathrm{MD}}(\theta)=1$ for every threshold
$\theta$, and therefore $D(0)=1$. Consequently,
\begin{equation*}
\lim_{s\downarrow0}h_L(s)=\ln D(0)=0=h_L(0).
\end{equation*}
Thus, $h_L(s)$ is continuous on $[0,\infty)$. This proves
Lemma~\ref{lem:log_dep_curvature}(a).

\smallskip
\subsection{Proof of part (b):}
For $s>0$,
\begin{equation}
h_L''(s) = \frac{D''(s)}{D(s)} - \left(\frac{D'(s)}{D(s)}\right)^{2}.
\label{eq:dep_log_second_derivative}
\end{equation}
As $s\downarrow0$, the Taylor expansions of $a(s)$ and $b(s)$ are
\begin{equation}
a(s) = L+\frac{L}{2}s-\frac{L}{6}s^{2}+O(s^{3})
\label{eq:a_taylor}
\end{equation}
and
\begin{equation}
b(s) = L-\frac{L}{2}s+\frac{L}{3}s^{2}+O(s^{3}).
\label{eq:b_taylor}
\end{equation}
Therefore, $
a'(s)  \longrightarrow \frac{L}{2},\
   b'(s)  \longrightarrow -\frac{L}{2},\
a''(s) \longrightarrow -\frac{L}{3},\
   b''(s) \longrightarrow \frac{2L}{3}.$
Let
\begin{equation}
\varphi_L(t) \triangleq \frac{t^{L-1}e^{-t}}{\Gamma(L)}, \qquad t>0.
\label{eq:dep_q}
\end{equation}
The derivatives of the incomplete gamma functions are
\begin{equation}
\frac{\partial}{\partial x}\gamma(L,x)
=
x^{L-1}e^{-x}
=
\Gamma(L)\varphi_L(x).
\label{eq:gamma_lower_derivative}
\end{equation}
and
\begin{equation}
\frac{\partial}{\partial x}\Gamma(L,x)
=
-x^{L-1}e^{-x}
=
-\Gamma(L)\varphi_L(x).
\label{eq:gamma_upper_derivative}
\end{equation}

\noindent Hence, differentiating~\eqref{eq:dep_upper_gamma} gives
\begin{equation}
D'(s)
=
-\varphi_L(a(s))a'(s)
+
\varphi_L(b(s))b'(s).
\label{eq:dep_D_prime}
\end{equation}
and
\begin{align}
D''(s)
={}&
-\varphi_L'(a(s))(a'(s))^2
-\varphi_L(a(s))a''(s)
\notag\\
&+
\varphi_L'(b(s))(b'(s))^2
+\varphi_L(b(s))b''(s).
\label{eq:dep_D_double_prime}
\end{align}

Substituting the limits obtained from
\eqref{eq:a_taylor}--\eqref{eq:b_taylor} into
\eqref{eq:dep_D_prime}, and using the continuity of $\varphi_L$ at $L$, gives
\begin{align}
\lim_{s\downarrow0}D'(s)
&=
-\varphi_L(L)\frac{L}{2}
+\varphi_L(L)\left(-\frac{L}{2}\right)
\notag\\
&=
-L\varphi_L(L).
\label{eq:D_prime_limit}
\end{align}
Similarly, substituting the same limits into
\eqref{eq:dep_D_double_prime}, and using the continuity of $\varphi_L$ and $\varphi_L'$
at $L$, gives
\begin{align}
\lim_{s\downarrow0}D''(s)
={}&
-\varphi_L'(L)\frac{L^2}{4}
-\varphi_L(L)\left(-\frac{L}{3}\right)
\notag\\
&+
\varphi_L'(L)\frac{L^2}{4}
+\varphi_L(L)\frac{2L}{3}
\notag\\
={}&
L\varphi_L(L),
\label{eq:D_double_prime_limit}
\end{align}
where the two terms containing $\varphi_L'(L)$ cancel.

Moreover,
\begin{equation}
L\varphi_L(L)
=
\frac{L^Le^{-L}}{\Gamma(L)}
=
K.
\label{eq:Lq_equals_K}
\end{equation}

Taking $s\downarrow0$ in
\eqref{eq:dep_log_second_derivative} and using
$\lim_{s\downarrow0}D(s)=1$ together with
\eqref{eq:D_prime_limit}--\eqref{eq:Lq_equals_K}, we obtain
\begin{equation}
h_L''(0^{+}) \triangleq \lim_{s\downarrow0}h_L''(s)
             = K-K^{2}
             = K(1-K).
\label{eq:h_double_prime_zero}
\end{equation}
This proves Lemma~\ref{lem:log_dep_curvature}(b).

\smallskip
\subsection{Proof of part (c):}
We obtain the large-$s$ behavior of $D(s)$ in
\eqref{eq:dep_upper_gamma} by expanding separately the
upper incomplete gamma term $\Gamma(L,a(s))$ and the
lower incomplete gamma term $\gamma(L,b(s))$. We then combine the two
expansions and differentiate the logarithm of the resulting expression.

As $s\to\infty$,
\begin{equation}
a(s)
=
L\left(1+\frac{1}{s}\right)\ln(1+s)
\sim L\ln s
\longrightarrow\infty,
\label{eq:a_large_s}
\end{equation}
whereas
\begin{equation}
b(s)
=
\frac{L\ln(1+s)}{s}
\sim
\frac{L\ln s}{s}
\longrightarrow0.
\label{eq:b_large_s}
\end{equation}

Given $L$, the large-argument expansion in
\cite[\href{https://dlmf.nist.gov/8.11.E2}{(8.11.2)}--
\href{https://dlmf.nist.gov/8.11.E3}{(8.11.3)}]{NIST:DLMF}
gives
\begin{equation}
\Gamma(L,x)
=
x^{L-1}e^{-x}
\left[
1+O\left(x^{-1}\right)
\right],
\qquad x\to\infty.
\label{eq:gamma_upper_standard}
\end{equation}
Similarly, using $\gamma(L,x)=\Gamma(L)-\Gamma(L,x)$, the series expansion in
\cite[\href{https://dlmf.nist.gov/8.7.E3}{(8.7.3)}]{NIST:DLMF}
gives
\begin{equation}
\gamma(L,x)
=
\frac{x^L}{L}
\left[
1+O(x)
\right],
\qquad x\to0.
\label{eq:gamma_lower_standard}
\end{equation}

We first expand the lower incomplete gamma term $\gamma(L,b(s))$. Substituting
\eqref{eq:b_large_s} into \eqref{eq:gamma_lower_standard} gives
\begin{equation}
\begin{aligned}
\gamma(L,b(s))
&=
\frac{b(s)^L}{L}
\left[
1+O\bigl(b(s)\bigr)
\right]\\
&=
L^{L-1}
\!\left(\frac{\ln(1+s)}{s}\right)^L
\!\!\left[
1+
O\left(\frac{\ln(1+s)}{s}\right)\!
\right].
\end{aligned}
\label{eq:gamma_lower_term}
\end{equation}

We next expand the upper incomplete gamma term $\Gamma(L,a(s))$. The definitions of
$a(s)$ and $b(s)$ give
\begin{equation}
a(s)
=
(1+s)b(s)
=
L\ln(1+s)+b(s).
\label{eq:a_b_relations}
\end{equation}
It follows from \eqref{eq:a_b_relations} that
\begin{equation}
\begin{aligned}
e^{-a(s)}
&=
e^{-L\ln(1+s)}e^{-b(s)}\\
&=
(1+s)^{-L}e^{-b(s)}.
\end{aligned}
\label{eq:exp_Q_identity}
\end{equation}
Combining \eqref{eq:a_b_relations} and \eqref{eq:exp_Q_identity}
yields
\begin{equation}
\begin{aligned}
a(s)^{L-1}e^{-a(s)}
&=
\bigl((1+s)b(s)\bigr)^{L-1}
(1+s)^{-L}e^{-b(s)}\\
&=
\frac{b(s)^{L-1}e^{-b(s)}}{1+s}.
\end{aligned}
\label{eq:upper_prefactor_identity}
\end{equation}
\noindent $\!\!$Substituting~\eqref{eq:upper_prefactor_identity} into
\eqref{eq:gamma_upper_standard} and using
\eqref{eq:a_large_s}--\eqref{eq:b_large_s} gives
\begin{equation}
\begin{aligned}
\Gamma(L,a(s))
={}&
a(s)^{L-1}e^{-a(s)}
\left[
1+O\left(\frac{1}{a(s)}\right)
\right]\\
={}&
L^{L-1}
\left(\frac{\ln(1+s)}{s}\right)^L
\frac{1}{\ln(1+s)}\\
&\times
\left[
1+
O\left(
\frac{1}{\ln(1+s)}
+
\frac{\ln(1+s)}{s}
\right)
\right].
\end{aligned}
\label{eq:gamma_upper_term}
\end{equation}

Define
\begin{equation}
\varepsilon_L(s)
\triangleq
\frac{\Gamma(L)}{L^{L-1}}
\left(\frac{s}{\ln(1+s)}\right)^{L}
D(s)
-1-\frac{1}{\ln(1+s)}.
\label{eq:rho_definition}
\end{equation}
The function $\varepsilon_L(s)$ is twice continuously differentiable for all
sufficiently large $s$. Rearranging~\eqref{eq:rho_definition} gives
\begin{equation}
D(s)
=
\frac{L^{L-1}}{\Gamma(L)}
\left(\frac{\ln(1+s)}{s}\right)^{L}
\left[
1+\frac{1}{\ln(1+s)}+\varepsilon_L(s)
\right].
\label{eq:dep_asymptotic}
\end{equation}

To differentiate $D(s)$ in~\eqref{eq:dep_asymptotic} twice, we first
derive the orders of $\varepsilon_L'(s)$ and
$\varepsilon_L''(s)$. Direct differentiation of the exact definitions
\eqref{eq:lem1_a}--\eqref{eq:lem1_b} gives
\begin{equation}
\begin{aligned}
a'(s)&=O(s^{-1}),
&
a''(s)&=O(s^{-2}),\\
b'(s)&=
O\left(\frac{\ln(1+s)}{s^{2}}\right),
&
b''(s)&=
O\left(\frac{\ln(1+s)}{s^{3}}\right).
\end{aligned}
\label{eq:ab_derivative_bounds}
\end{equation}
Differentiating the exact definition in \eqref{eq:rho_definition}, with
$D'(s)$ and $D''(s)$ given by \eqref{eq:dep_D_prime}--\eqref{eq:dep_D_double_prime}, and estimating the
resulting expressions using \eqref{eq:a_large_s}--\eqref{eq:gamma_upper_term} and~\eqref{eq:ab_derivative_bounds}, gives,
\begin{equation}
\begin{aligned}
\varepsilon_L(s)
&=
O\left(
\frac{1}{\ln^{2}(1+s)}
+
\frac{\ln(1+s)}{s}
\right),\\
\varepsilon_L'(s)
&=
O\left(
\frac{1}{s\ln^{3}(1+s)}
+
\frac{\ln(1+s)}{s^{2}}
\right),\\
\varepsilon_L''(s)
&=
O\left(
\frac{1}{s^{2}\ln^{3}(1+s)}
+
\frac{\ln(1+s)}{s^{3}}
\right).
\end{aligned}
\label{eq:rho_derivative_bounds}
\end{equation}
The first bound in~\eqref{eq:rho_derivative_bounds} implies that
$\varepsilon_L(s)\to0$. Hence,
\begin{equation}
1+\frac{1}{\ln(1+s)}+\varepsilon_L(s)
\longrightarrow 1.
\label{eq:rho_log_argument_limit}
\end{equation}
Differentiating twice and using
\eqref{eq:rho_derivative_bounds}--\eqref{eq:rho_log_argument_limit}
yields
\begin{equation}
\frac{d^{2}}{ds^{2}}
\ln\left[
1+\frac{1}{\ln(1+s)}+\varepsilon_L(s)
\right]
=
o(s^{-2}).
\label{eq:log_remainder_second_derivative}
\end{equation}

Finally, taking the logarithm of both sides of
\eqref{eq:dep_asymptotic} gives
\begin{equation}
\begin{aligned}
h_L(s)
={}&
\ln\left(\frac{L^{L-1}}{\Gamma(L)}\right)
+
L\ln\bigl(\ln(1+s)\bigr)
-
L\ln s\\
&+
\ln\left[
1+\frac{1}{\ln(1+s)}+\varepsilon_L(s)
\right].
\end{aligned}
\label{eq:h_asymptotic_expansion}
\end{equation}
The first term in~\eqref{eq:h_asymptotic_expansion} is independent
of $s$. Direct differentiation of
$\ln\bigl(\ln(1+s)\bigr)$ gives
\begin{equation}
\frac{d^{2}}{ds^{2}}
\ln\bigl(\ln(1+s)\bigr)
=
-\frac{\ln(1+s)+1}
{(1+s)^{2}\ln^{2}(1+s)}
=
o(s^{-2}).
\label{eq:loglog_second_derivative}
\end{equation}
Differentiating~\eqref{eq:h_asymptotic_expansion} twice and using
\eqref{eq:log_remainder_second_derivative} and
\eqref{eq:loglog_second_derivative} gives
\begin{equation}
\begin{aligned}
h_L''(s)
&=
\frac{L}{s^{2}}
-
\frac{L\bigl(\ln(1+s)+1\bigr)}
{(1+s)^{2}\ln^{2}(1+s)}
+
o(s^{-2})\\
&=
\frac{L}{s^{2}}
+
o(s^{-2})
\sim
\frac{L}{s^{2}}.
\end{aligned}
\label{eq:h_double_prime_asymptotic}
\end{equation}
Because $L>0$, \eqref{eq:h_double_prime_asymptotic} implies that
$h_L''(s)>0$ for sufficiently large $s$. Therefore, $h_L(s)$ is
convex for sufficiently large $s$. This proves part~(c) and completes
the proof of Lemma~\ref{lem:log_dep_curvature}.

\section{Equivalence and Four-Block Formulation}
\label{app:fb-formulation}
This appendix has four parts.
Part A proves that the link-level constraints
\eqref{eq:P1_cap}--\eqref{eq:P1_sched} are equivalent to the endpoint-local
constraints \eqref{eq:consensus}--\eqref{eq:local_exclusive}.
Part B groups the primal variables of Problems
\eqref{prob:P1} and \eqref{prob:unified} into four blocks and writes the equalities
\eqref{eq:P1_flow}, \eqref{eq:P1_rate}, \eqref{eq:p1-snragg}, and \eqref{eq:consensus} in matrix form.
Part C constructs the proximal matrices for the
local primal updates \eqref{eq:step_x}, \eqref{eq:routing}, \eqref{eq:nodeupdate}, and \eqref{eq:s_update}.
Part D proves that the four-block updates
\eqref{eq:fb-primal-update}--\eqref{eq:fb-dual-update}
generate the same primal--dual iterates as these local
primal updates and the dual updates \eqref{eq:lamupd}-\eqref{eq:omegaupd}.

\subsection{Equivalence of the Original Constraints \eqref{eq:P1_cap}--\eqref{eq:P1_sched} and the Endpoint-Local Constraints \eqref{eq:consensus}--\eqref{eq:local_exclusive}:}
For each $n\in\mathcal N$ and $m\in\mathcal M$, let
$\mathbf y_m^{(n)}
\triangleq
\big(y_{ml}^{(n)}\big)_{
l\in(\mathcal I(n)\cup\mathcal O(n))\cap\mathcal L_m}$,
and let $\mathcal S_{mn}$ denote the local scheduling set
specified by \eqref{eq:local_exclusive}.
\begin{proposition}[Equivalence of \eqref{eq:P1_cap}--\eqref{eq:P1_sched} and
\eqref{eq:consensus}--\eqref{eq:local_exclusive}]
\label{prop:local_equivalence}
For both Problem~\eqref{prob:P1} and Problem~\eqref{prob:unified}, replacing
\eqref{eq:P1_cap}--\eqref{eq:P1_sched} with \eqref{eq:consensus}--\eqref{eq:local_exclusive}, while retaining the remaining
constraints, does not change the feasible network variables or
the objective value. At every feasible point, the scheduling
variable is recovered as
$y_{ml}=y_{ml}^{(i)}=y_{ml}^{(j)}$.
\end{proposition}

\begin{IEEEproof}
We show this equivalence in both directions. Step 1 constructs the local
scheduling copies from a feasible point of Problem~\eqref{prob:P1} or~\eqref{prob:unified}.
Step 2 recovers the variables satisfying \eqref{eq:P1_cap}--\eqref{eq:P1_sched} from
variables satisfying \eqref{eq:consensus}--\eqref{eq:local_exclusive}.

\subsubsection{Step 1: From \eqref{eq:P1_cap}--\eqref{eq:P1_sched} to \eqref{eq:consensus}--\eqref{eq:local_exclusive}}
Consider any feasible solution of Problem~\eqref{prob:P1} or~\eqref{prob:unified}. Keep
$x_f$, $r_{ml}^{d}$, $p_{ml}$, $q_{ml}$, and $s_{w,m}$
unchanged. For every $l=(i,j)\in\mathcal L$ and
$m\in\mathcal M_l$, define
\begin{equation}
y_{ml}^{(i)}
\triangleq
y_{ml}^{(j)}
\triangleq
y_{ml}.
\label{eq:local_copy_construction}
\end{equation}
The local copies in \eqref{eq:local_copy_construction} satisfy the consensus constraint
\eqref{eq:consensus}. Since the original scheduling variables satisfy \eqref{eq:P1_sched},
we also have
$\mathbf y_m^{(n)}\in\mathcal S_{mn}$ for every node $n$ and
modality $m$, which gives \eqref{eq:local_exclusive}.

Moreover, \eqref{eq:P1_rate} and \eqref{eq:P1_rbnd} give $q_{ml}\geq0$. Together with
\eqref{eq:P1_cap} and \eqref{eq:local_copy_construction}, this gives
\begin{equation}
0\leq q_{ml}
\leq
C_{ml}\!\left(p_{ml},y_{ml}^{(i)}\right),
\end{equation}
which is \eqref{eq:local_cap}. Hence, the constructed variables satisfy
\eqref{eq:consensus}--\eqref{eq:local_exclusive}.

\subsubsection{Step 2: From \eqref{eq:consensus}--\eqref{eq:local_exclusive} to \eqref{eq:P1_cap}--\eqref{eq:P1_sched}}
Consider any variables satisfying \eqref{eq:consensus}--\eqref{eq:local_exclusive} and the
remaining constraints of Problem~\eqref{prob:P1} or~\eqref{prob:unified}. For every $l=(i,j)\in\mathcal L$ and $m\in\mathcal M_l$, define
\(
y_{ml}
\triangleq
y_{ml}^{(i)}
=
y_{ml}^{(j)},
\)
where the equality follows from \eqref{eq:consensus}.
Constraints \eqref{eq:P1_rate} and
\eqref{eq:local_cap} then give
\begin{equation}
\!\!\sum_{d\in\mathcal N}r_{ml}^{d}
=
q_{ml}
\leq
C_{ml}(p_{ml},y_{ml}),
\quad
\forall\,l\in\mathcal L,\ m\in\mathcal M_l,
\label{eq:capacity_recovery}
\end{equation}
which recovers \eqref{eq:P1_cap}. Moreover, \eqref{eq:local_exclusive} gives
\begin{align}
\!&\sum_{l\in(\mathcal I(n)\cup\mathcal O(n))
\cap\mathcal L_m}
\!y_{ml}
 =\!\!
\sum_{l\in(\mathcal I(n)\cup\mathcal O(n))
\cap\mathcal L_m}
\!y_{ml}^{(n)}
\leq1,
\nonumber\\[-1mm]
&\hspace{33mm}
\forall\,n\in\mathcal N,\ m\in\mathcal M,
\label{eq:scheduling_recovery}
\end{align}
which recovers \eqref{eq:P1_sched}. All remaining constraints are unchanged.

Finally, $q_{ml}$ and the local scheduling copies do not appear
in either objective. Therefore, the corresponding feasible
points have the same objective value, and an optimal solution
of either representation gives an optimal solution of the
other.
\end{IEEEproof}

\subsection{Equivalent Four-Block Reformulation
of Problems \eqref{prob:P1} and \eqref{prob:unified}:}

We group the primal variables
$\mathbf x,\mathbf r,\mathbf p,\mathbf q,
\mathbf y^{(n)},\mathbf s$
of the endpoint-local Problems \eqref{prob:P1} and \eqref{prob:unified}
into four blocks.
We then write the equalities
\eqref{eq:P1_flow}, \eqref{eq:P1_rate}, \eqref{eq:p1-snragg}, and \eqref{eq:consensus} in matrix form.

Let $\mathbf y$ collect all endpoint-local scheduling
copies. Define the four primal blocks by
\begin{equation}
\begin{aligned}
\boldsymbol\xi_1&=\mathbf x,
&\boldsymbol\xi_2&=\mathbf r,\\
\boldsymbol\xi_3&=\mathbf h
\triangleq
\begin{bmatrix}
\mathbf p\\
\mathbf q\\
\mathbf y
\end{bmatrix},
&\boldsymbol\xi_4&=\mathbf s.
\end{aligned}
\label{eq:fb-blocks}
\end{equation}
The four blocks contain the congestion-control, routing,
joint power--rate--scheduling, and aggregate-SNR variables,
respectively. Collect the primal and dual variables as
\begin{equation}
\begin{aligned}
\boldsymbol\xi
&\triangleq
\begin{bmatrix}
\boldsymbol\xi_1\\
\boldsymbol\xi_2\\
\boldsymbol\xi_3\\
\boldsymbol\xi_4
\end{bmatrix},
&
\boldsymbol\Lambda
&\triangleq
\begin{bmatrix}
\boldsymbol\lambda\\
\boldsymbol\mu\\
\boldsymbol\chi\\
\boldsymbol\omega
\end{bmatrix},\\
\mathbf z
&\triangleq
\begin{bmatrix}
\boldsymbol\xi\\
\boldsymbol\Lambda
\end{bmatrix}.
\end{aligned}
\label{eq:fb-state}
\end{equation}

Using the residual signs in \eqref{eq:auglag}, constraints
\eqref{eq:P1_flow}, \eqref{eq:P1_rate}, \eqref{eq:p1-snragg}, and \eqref{eq:consensus} are equivalent to
$\mathbf c(\boldsymbol\xi)=0$, where
\begin{equation}
\mathbf c(\boldsymbol\xi)\triangleq
\begin{bmatrix}
\left(
\begin{aligned}
&\sum_{l\in O(n)}\sum_{m\in\mathcal M_l}r_{ml}^{d}-\sum_{l\in I(n)}\sum_{m\in\mathcal M_l}r_{ml}^{d}\\
&-\sum_{f\in\mathcal F}
  \mathbf1_{\{s_f=n,\,d_f=d\}}x_f
\end{aligned}
\right)_{\substack{n,d\in\mathcal N\\n\ne d}}
\\[2mm]
\bigg(
\displaystyle\sum_{d\in\mathcal N}r_{ml}^{d}-q_{ml}
\bigg)_{{l\in\mathcal L,\ m\in\mathcal M_l}}
\\[2mm]
\bigg(
s_{w,m}
-\displaystyle\sum_{l\in\mathcal L_m}A_{w,ml}p_{ml}
\bigg)_{{m\in\mathcal M,\ w\in\mathcal W_m}}
\\[2mm]
\bigg(
y_{ml}^{(i)}-y_{ml}^{(j)}
\bigg)_{{l=(i,j)\in\mathcal L,\ m\in\mathcal M_l}}
\end{bmatrix}.
\label{eq:fb-equality-residual}
\end{equation}
Each indexed expression denotes a column vector over
the indicated indices. We use the same fixed ordering
of these indices throughout the formulation and proofs.

The residual in~\eqref{eq:fb-equality-residual} is linear
in the four primal blocks:
\begin{equation}
\begin{aligned}
\mathbf c(\boldsymbol\xi)
&=\sum_{i=1}^{4}Q_i\boldsymbol\xi_i
=Q\boldsymbol\xi,\quad
Q=[\,Q_1\ Q_2\ Q_3\ Q_4\,].
\end{aligned}
\label{eq:fb-equalities}
\end{equation}
The following identities, for arbitrary block vectors
of the corresponding dimensions, define the fixed
coefficient matrices $Q_i$:
\begin{equation}
\begin{aligned}
Q_1\mathbf x
&=
\begin{bmatrix}
\bigg(
-\displaystyle\sum_{{f\in\mathcal F,\
s_f=n,\ d_f=d}}x_f
\bigg)_{{n,\ d\in\mathcal N,\ n\ne d}}\\
0\\
0\\
0
\end{bmatrix},
\\[2mm]
Q_2\mathbf r
&=
\begin{bmatrix}
\left(
\begin{aligned}
&\sum_{l\in O(n)}\sum_{m\in\mathcal M_l}r_{ml}^{d}\\&-\sum_{l\in I(n)}\sum_{m\in\mathcal M_l}r_{ml}^{d}
\end{aligned}
\right)_{n,\ d\in\mathcal N,\ n\ne d}
\\[1mm]
\bigg(
\displaystyle\sum_{d\in\mathcal N}r_{ml}^{d}
\bigg)_{{l\in\mathcal L,\ m\in\mathcal M_l}}\\
0\\
0
\end{bmatrix},
\\[2mm]
Q_3
\begin{bmatrix}
\mathbf p\\
\mathbf q\\
\mathbf y
\end{bmatrix}
&=
\begin{bmatrix}
0\\
(-q_{ml})_{{l\in\mathcal L,\ m\in\mathcal M_l}}
\\[1mm]
\bigg(
-\displaystyle\sum_{l\in\mathcal L_m}A_{w,ml}p_{ml}
\bigg)_{{m\in\mathcal M,\ w\in\mathcal W_m}}
\\[1mm]
\bigg(
y_{ml}^{(i)}-y_{ml}^{(j)}
\bigg)_{{l=(i,j)\in\mathcal L,\ m\in\mathcal M_l}}
\end{bmatrix},
\\[2mm]
Q_4\mathbf s
&=
\begin{bmatrix}
0\\
0\\
(s_{w,m})_{\substack{m\in\mathcal M\\w\in\mathcal W_m}}\\
0
\end{bmatrix}.
\end{aligned}
\label{eq:fb-A}
\end{equation}
The four block rows follow the order in
\eqref{eq:fb-equality-residual}.
Each zero denotes a zero vector of the corresponding
dimension. The columns of $Q_3$ follow the order
$(\mathbf p,\mathbf q,\mathbf y)$.

The remaining constraints are imposed within each block.
Let $\mathcal X$ be the admitted-rate box in (1), and let
$\mathcal R$ be the set of nonnegative routing vectors.
Let $\mathcal H$ be the product of the node--modality
feasible sets in \eqref{eq:nodeupdate}, with coordinates ordered as
$(\mathbf p,\mathbf q,\mathbf y)$.
Let $\mathcal S_{\mathrm{snr}}$ be the product of the
scalar feasible sets in \eqref{eq:s_update}. Define
\begin{equation}
(\mathcal C_1,\mathcal C_2,\mathcal C_3,\mathcal C_4)
=(\mathcal X,\mathcal R,\mathcal H,\mathcal S_{\mathrm{snr}}).
\label{eq:fb-local-sets}
\end{equation}

These sets are nonempty under Assumption 1 and are closed
and convex. In particular, the capacity constraints
$q_{ml}\le C_{ml}(p_{ml},y_{ml}^{(i)})$ are convex because
$C_{ml}$ is continuous and concave. Similarly,
$g_{L_{w,m}}(s_{w,m})\ge\ln\eta_{w,m}$ defines a closed
convex set because $g_{L_{w,m}}$ is continuous and concave.
The remaining local constraints are linear.

To include the local constraints in the objective,
define the indicator function of a set $\mathcal C$ by
\begin{equation}
\iota_{\mathcal C}(\mathbf v)
\triangleq
\begin{cases}
0, & \mathbf v\in\mathcal C,\\
+\infty, & \mathbf v\notin\mathcal C.
\end{cases}
\label{eq:fb-indicator}
\end{equation}
For the local feasible variables, define
\begin{equation}
\begin{aligned}
f_1(\mathbf x)
 &=-\sum_{f\in\mathcal F}U_f(x_f),\\
f_2(\mathbf r)&=0,\qquad f_3(\mathbf h)=0,\\
f_4(\mathbf s)
 &=-\sum_{m\in\mathcal M}\sum_{w\in\mathcal W_m}
 G_w\bigl(g_{L_{w,m}}(s_{w,m})\bigr).
\end{aligned}
\label{eq:fb-smooth-objectives}
\end{equation}
The functions $f_i$ are convex on their respective sets.
For $f_1$, this follows from the concavity of $U_f$.
For $f_4$, it follows because $G_w$ is concave and
nondecreasing and $g_{L_{w,m}}$ is concave.

Let
\begin{equation}
\Phi_i(\mathbf v)
\triangleq
f_i(\mathbf v)+\iota_{\mathcal C_i}(\mathbf v),
\qquad i=1,\ldots,4,
\label{eq:fb-block-objectives}
\end{equation}
with $\Phi_i(\mathbf v)=+\infty$ outside $\mathcal C_i$.
Since each $f_i$ is finite, continuous, and convex on the
nonempty closed convex set $\mathcal C_i$, each $\Phi_i$
is proper, closed, and convex.

Minimizing the negative of the utility in Problem \eqref{prob:unified}
therefore gives
\begin{equation}
\begin{aligned}
\underset{\boldsymbol\xi_1,\ldots,\boldsymbol\xi_4}
 {\operatorname{minimize}}\quad&
 \sum_{i=1}^{4}\Phi_i(\boldsymbol\xi_i)\\
\operatorname{subject~to}\quad&
 \sum_{i=1}^{4}Q_i\boldsymbol\xi_i=0.
\end{aligned}
\label{eq:fb-problem}
\end{equation}
This is the standard separable convex form for multi-block
ADMM~\cite{deng2017parallel}, with four blocks.
The objective is a sum of four convex functions, and the
displayed constraints are linear equalities.
The local constraints are included through the indicator
functions in $\Phi_i$.
For Problem \eqref{prob:P1}, set $f_4=0$.

We use the multiplier term
$\langle\boldsymbol\Lambda,Q\boldsymbol\xi\rangle$
in this minimization formulation.
This sign agrees with the dual updates \eqref{eq:lamupd}--\eqref{eq:omegaupd}.

\subsection{Proximal Matrices for the Four Primal Blocks $\mathbf x$, $\mathbf r$,
$(\mathbf p,\mathbf q,\mathbf y^{(n)})$, and $\mathbf s$:}

We choose the proximal matrices so that the four-block
primal update~\eqref{eq:fb-primal-update} gives the local
updates \eqref{eq:step_x}, \eqref{eq:routing}, \eqref{eq:nodeupdate}, and \eqref{eq:s_update}.
First, define the matrices $D_i$ from the squared terms
of \eqref{eq:auglag} appearing in the primal updates
\eqref{eq:step_x}, \eqref{eq:routing}, \eqref{eq:nodeupdate}, and \eqref{eq:s_update}, after factoring out $\rho/2$.
For the routing variable $r_{ml}^{d}$ on a directed
link $l=(i,j)$, set
\begin{equation}
\begin{aligned}
&D_1=I,\qquad D_4=I,\\
&D_3=
\begin{bmatrix}
D_p&0&0\\
0&I_q&0\\
0&0&I_y
\end{bmatrix},\\
&[D_2]_{(m,l,d),(m,l,d)}
=1+\mathbf1_{\{i\ne d\}}+\mathbf1_{\{j\ne d\}}.
\end{aligned}
\label{eq:fb-D}
\end{equation}
All off-diagonal entries of $D_2$ are zero, and its
rows and columns follow the ordering of $\mathbf r$.
The rows and columns of $D_3$ follow the order
$(\mathbf p,\mathbf q,\mathbf y)$, and $D_p$ is
defined in~\eqref{eq:fb-Dp}.
The identity matrices in $D_1$ and $D_4$ have the
dimensions of $\mathbf x$ and $\mathbf s$, respectively;
$I_q$ and $I_y$ have the dimensions of $\mathbf q$
and $\mathbf y$, respectively.

Each admitted-rate variable appears in one
flow-conservation equality, which gives $D_1=I$.
For each diagonal entry of $D_2$, the term $1$ comes
from the aggregate-rate equality \eqref{eq:P1_rate}. The remaining
two terms come from the flow-conservation
equalities \eqref{eq:P1_flow} at the endpoints of the link.


Let $\operatorname{tx}(l)$ denote the transmitter of link $l$.
The power matrix is
\begin{equation}
\!\!\![D_p]_{ml,m'l'}=
\!\begin{cases}
\displaystyle\sum_{w\in\mathcal W_m}\!A_{w,ml}A_{w,ml'},
&\begin{array}{l}
\!\!\!m=m',\\
\!\!\!\operatorname{tx}(l)=\operatorname{tx}(l'),
\end{array}\\
0,&\text{otherwise}.
\end{cases}
\label{eq:fb-Dp}
\end{equation}
Thus, $D_p$ retains the cross terms between powers
updated jointly by one node on one modality.
Entries between different node--modality updates are zero.
Each aggregate rate and each endpoint scheduling copy
appears in one equality row, which gives $I_q$ and $I_y$.
Each aggregate-SNR variable also appears in one equality
row, which gives $D_4=I$.
The capacity constraints remain in $\mathcal H$ and are
not part of the squared equality penalties in \eqref{eq:auglag}.

To express the connection with the local updates, let
$\mathbf u_b$ contain the variables in one local update:
one $x_f$, one $r_{ml}^{d}$, one node--modality vector
in \eqref{eq:nodeupdate}, or one $s_{w,m}$.
Let $Q_b^{\mathrm{loc}}$ be the submatrix of $Q$
containing the columns corresponding to the variables
in $\mathbf u_b$, in the same order as in $\mathbf u_b$.
Then $Q_b^{\mathrm{loc}}\mathbf u_b$ is their contribution
to $Q\boldsymbol\xi$.
Let $\mathcal I_i$ index the local updates contained
in block $i$.
For an increment $\mathbf d_i$ with local components
$\mathbf d_b$,
\begin{equation}
\begin{aligned}
Q_i\mathbf d_i
 &=\sum_{b\in\mathcal I_i}Q_b^{\mathrm{loc}}\mathbf d_b,\\
\mathbf d_i^{\mathsf T}D_i\mathbf d_i
 &=\sum_{b\in\mathcal I_i}\|Q_b^{\mathrm{loc}}\mathbf d_b\|^2.
\end{aligned}
\label{eq:fb-local-identity}
\end{equation}
The second identity also shows that $D_i\succeq0$.

Define the four proximal matrices by
\begin{equation}
P_i\triangleq
\alpha I+\rho D_i-\rho Q_i^{\mathsf T}Q_i,
\qquad i=1,\ldots,4.
\label{eq:fb-P}
\end{equation}
They satisfy
\begin{equation}
P_i+\rho Q_i^{\mathsf T}Q_i
=\alpha I+\rho D_i.
\label{eq:fb-cancellation}
\end{equation}
This identity makes the total quadratic term of each
four-block update equal to the sum of its local
quadratic terms, as shown next.
For a symmetric matrix $M$, write
$\|\mathbf v\|_M^2=\mathbf v^{\mathsf T}M\mathbf v$.
Under Assumption \ref{assump:proximal}, each $P_i$ is positive definite,
as shown in Appendix~\ref{app:fb-descent}.

\subsection{Equivalence of Four-Block PP-ADMM
to Updates \eqref{eq:step_x}, \eqref{eq:routing}, \eqref{eq:nodeupdate}, \eqref{eq:s_update}, and \eqref{eq:lamupd}--\eqref{eq:omegaupd}:}

We now show that the four-block primal and dual updates
\eqref{eq:fb-primal-update}--\eqref{eq:fb-dual-update},
with the proximal matrices in~\eqref{eq:fb-P},
give the local primal updates \eqref{eq:step_x}, \eqref{eq:routing}, \eqref{eq:nodeupdate}, and \eqref{eq:s_update}
and the dual updates \eqref{eq:lamupd}--\eqref{eq:omegaupd}.

For the standard problem~\eqref{eq:fb-problem}, the
four primal updates are
\begin{equation}
\begin{aligned}
\boldsymbol\xi_i^{k+1}
=\arg\min_{\mathbf v}\ &\Bigg\{
\Phi_i(\mathbf v)
+\frac12\|\mathbf v-\boldsymbol\xi_i^k\|_{P_i}^2\\
&+\frac{\rho}{2}
\left\|Q_i\mathbf v+
\sum_{j\ne i}Q_j\boldsymbol\xi_j^k
+\frac{\boldsymbol\Lambda^k}{\rho}\right\|^2
\Bigg\},
\end{aligned}
\label{eq:fb-primal-update}
\end{equation}
for $i=1,\ldots,4$, followed by
\begin{equation}
\boldsymbol\Lambda^{k+1}
=\boldsymbol\Lambda^k+\tau\rho Q\boldsymbol\xi^{k+1}.
\label{eq:fb-dual-update}
\end{equation}
These are the proximal parallel ADMM updates
in~\cite{deng2017parallel}, with four primal blocks.
All four primal updates use iteration-$k$ values of
the other blocks and can therefore be computed in parallel.

Fix a block $i$ and write
\[
\mathbf a^k
=Q\boldsymbol\xi^k+\boldsymbol\Lambda^k/\rho,
\qquad
\mathbf d_i=\mathbf v-\boldsymbol\xi_i^k.
\]
By~\eqref{eq:fb-local-identity} and
\eqref{eq:fb-cancellation}, the quadratic terms in
\eqref{eq:fb-primal-update} satisfy
\begin{equation}
\begin{aligned}
&\frac{\rho}{2}\|\mathbf a^k+Q_i\mathbf d_i\|^2
+\frac12\mathbf d_i^{\mathsf T}P_i\mathbf d_i\\
&=
\frac{\rho}{2}\|\mathbf a^k\|^2
+\sum_{b\in\mathcal I_i}
\Bigl[
\rho\langle\mathbf a^k,Q_b^{\mathrm{loc}}\mathbf d_b\rangle\\
&\hspace{30mm}
+\frac{\rho}{2}\|Q_b^{\mathrm{loc}}\mathbf d_b\|^2
+\frac{\alpha}{2}\|\mathbf d_b\|^2
\Bigr].
\end{aligned}
\label{eq:fb-expansion}
\end{equation}
The first term in~\eqref{eq:fb-expansion} is constant
in the update variables.
Each term in the sum depends on the variables of one
local update and equals its quadratic terms in
\eqref{eq:step_x}, \eqref{eq:routing}, \eqref{eq:nodeupdate}, or \eqref{eq:s_update}, up to a constant.

The objective and feasible set also separate across
these local updates. Hence, minimizing over an entire
block is equivalent to solving its local updates
separately. Since $\alpha>0$, each local minimization
has a unique solution, so the primal updates coincide.

The primal variables produced by the four-block
update~\eqref{eq:fb-primal-update} therefore equal those
produced by \eqref{eq:step_x}, \eqref{eq:routing}, \eqref{eq:nodeupdate}, and \eqref{eq:s_update}.
Their equality residuals
$Q\boldsymbol\xi^{k+1}$ in~\eqref{eq:fb-equalities}
are consequently identical.
The four-block dual update~\eqref{eq:fb-dual-update}
then agrees with \eqref{eq:lamupd}--\eqref{eq:omegaupd}.
Starting from the same initial point and using the
same parameters, these updates generate the same
primal--dual iterates at every iteration.

\section{Proof of Theorem~\ref{thm:linear_convergence}}
\label{apx:linear_convergence}
\subsection{Notation and KKT Residual}
We first define the KKT residual and the matrix $G$
used in the proof. 
Let
\begin{equation}
\mathcal D
\triangleq
\left\{\mathbf z:
\boldsymbol\xi_i\in\mathcal C_i,\ i=1,\ldots,4\right\},
\label{eq:fb-domain}
\end{equation}
where the dual coordinates are unrestricted.
The coupling equalities need not hold at a point
in $\mathcal D$.

For a nonempty closed convex set $\mathcal C$, define
the Euclidean projection by
\begin{equation}
\Pi_{\mathcal C}(\mathbf u)
\triangleq
\arg\min_{\mathbf v\in\mathcal C}\|\mathbf u-\mathbf v\|^2.
\label{eq:fb-projection}
\end{equation}
The minimizer is unique. Moreover, for $\mathbf u\in\mathcal C$, the equality
$\mathbf u=\Pi_{\mathcal C}(\mathbf u+\mathbf v)$
holds if and only if
\begin{equation}
\langle\mathbf v,\mathbf w-\mathbf u\rangle\le0
\qquad\text{for all }\mathbf w\in\mathcal C.
\label{eq:fb-projection-condition}
\end{equation}

The projected residual uses the gradients of the utility
terms and the quadratic penalty in \eqref{eq:auglag}.
These gradients are continuous on the local domains
by the assumptions on $U_f$ and $G_w$ and the matching
derivatives in \eqref{eq:g_definition}.
The capacity constraints remain in $\mathcal C_3$
and are not differentiated.

For $\mathbf z\in\mathcal D$, define the KKT residual
$\mathbf R(\mathbf z)$ of Problem~\eqref{eq:fb-problem} by
\begin{equation}
\mathbf R(\mathbf z)\triangleq
\begin{bmatrix}
\mathbf x-\Pi_{\mathcal X}\!\left(
\mathbf x+\nabla_{\mathbf x}\mathcal L_\rho(\mathbf z)
\right)
\\[2mm]
\mathbf r-\Pi_{\mathcal R}\!\left(
\mathbf r+\nabla_{\mathbf r}\mathcal L_\rho(\mathbf z)
\right)
\\[2mm]
\begin{bmatrix}
\mathbf p\\
\mathbf q\\
\mathbf y
\end{bmatrix}
-\Pi_{\mathcal H}\!\left(
\begin{bmatrix}
\mathbf p\\
\mathbf q\\
\mathbf y
\end{bmatrix}
+
\begin{bmatrix}
\nabla_{\mathbf p}\mathcal L_\rho(\mathbf z)\\
\nabla_{\mathbf q}\mathcal L_\rho(\mathbf z)\\
\nabla_{\mathbf y}\mathcal L_\rho(\mathbf z)
\end{bmatrix}
\right)
\\[2mm]
\mathbf s-\Pi_{\mathcal S_{\mathrm{snr}}}\!\left(
\mathbf s+\nabla_{\mathbf s}\mathcal L_\rho(\mathbf z)
\right)
\\[2mm]
\left(
\begin{aligned}
&\sum_{l\in O(n)}\sum_{m\in\mathcal M_l}r_{ml}^{d}-\sum_{l\in I(n)}\sum_{m\in\mathcal M_l}r_{ml}^{d}\\
&-\sum_{f\in\mathcal F}
  \mathbf1_{\{s_f=n,\,d_f=d\}}x_f
\end{aligned}
\right)_{\substack{n,d\in\mathcal N\\n\ne d}}
\\[2mm]
\bigg(
\displaystyle\sum_{d\in\mathcal N}r_{ml}^{d}-q_{ml}
\bigg)_{{l\in\mathcal L,\ m\in\mathcal M_l}}
\\[2mm]
\bigg(
s_{w,m}
-\displaystyle\sum_{l\in\mathcal L_m}A_{w,ml}p_{ml}
\bigg)_{m\in\mathcal M,\ w\in\mathcal W_m}
\\[2mm]
\bigg(
y_{ml}^{(i)}-y_{ml}^{(j)}
\bigg)_{l=(i,j)\in\mathcal L,\ m\in\mathcal M_l}
\end{bmatrix}.
\label{eq:fb-kkt-residual}
\end{equation}
The feasible sets are defined in~\eqref{eq:fb-local-sets},
and $\mathcal L_\rho$ is given in \eqref{eq:auglag}.
The first four block rows are the projection differences
for $\mathbf x$, $\mathbf r$,
$(\mathbf p,\mathbf q,\mathbf y)$, and $\mathbf s$.
The last four block rows are the residuals of
\eqref{eq:P1_flow}, \eqref{eq:P1_rate}, \eqref{eq:p1-snragg}, and \eqref{eq:consensus}, respectively, with the
same ordering as in~\eqref{eq:fb-equality-residual}.

By \eqref{eq:fb-projection-condition}, each of the first four block rows is zero
exactly when the corresponding block satisfies the
first-order optimality condition for maximizing \eqref{eq:auglag}
over its feasible set, with the other blocks fixed.
When the last four block rows are also zero, these
conditions are the KKT conditions of
Problem~\eqref{eq:fb-problem}.
Convexity and Assumption 1 make these conditions
necessary and sufficient for primal--dual optimality.
Hence, for $\mathbf z\in\mathcal D$,
$\mathbf R(\mathbf z)=0$ holds if and only if
$\mathbf z\in\mathcal Z^*$.

For each primal group, define
\begin{equation}
H_i
\triangleq P_i+\rho Q_i^{\mathsf T}Q_i
=\alpha I+\rho D_i,
\qquad i=1,\ldots,4.
\label{eq:fb-H}
\end{equation}
Define the primal matrix $G_\xi$ and the complete
Lyapunov matrix $G$ by
\begin{equation}
\begin{aligned}
G_{\xi}
&\triangleq
\begin{bmatrix}
H_1 & 0   & 0   & 0\\
0   & H_2 & 0   & 0\\
0   & 0   & H_3 & 0\\
0   & 0   & 0   & H_4
\end{bmatrix},\\
G
&\triangleq
\begin{bmatrix}
G_{\xi} & 0\\
0 & \dfrac{1}{\rho\tau}I
\end{bmatrix}.
\end{aligned}
\label{eq:fb-G}
\end{equation}
Since $D_i\succeq0$ and $\alpha,\rho,\tau>0$,
we have $G\succ0$.

The Lyapunov function in~\eqref{eq:lyapunov} is
\begin{equation}
\begin{aligned}
V(\mathbf z;\mathbf z^*)
&=
\frac{1}{\rho\tau}
\|\boldsymbol\Lambda-\boldsymbol\Lambda^*\|^2+
\sum_{i=1}^{4}
\|\boldsymbol\xi_i-\boldsymbol\xi_i^*\|_{H_i}^2\\
&=\|\mathbf z-\mathbf z^*\|_G^2.
\end{aligned}
\label{eq:fb-expanded-lyapunov}
\end{equation}
By~\eqref{eq:fb-local-identity},
\begin{equation}
\begin{aligned}
&\sum_{i=1}^{4}
\|\boldsymbol\xi_i-\boldsymbol\xi_i^*\|_{H_i}^2\\
&=
\alpha\|\boldsymbol\xi-\boldsymbol\xi^*\|^2
+\rho\sum_{i=1}^{4}\sum_{b\in\mathcal I_i}
\|Q_b^{\mathrm{loc}}(\mathbf u_b-\mathbf u_b^*)\|^2.
\end{aligned}
\label{eq:fb-metric-identity}
\end{equation}
Thus, the four-block expression collects the same local
quadratic terms without changing their values.

\subsection{Proof of Theorem~\ref{thm:linear_convergence}}
The proof proceeds in four steps. Step 1 extends the local KKT error bound \eqref{eq:kkt_error_bound} to the
uniform bound on any compact set $\mathcal K\subseteq\mathcal D$ containing a saddle point.
Step 2 establishes sufficient decrease in the Lyapunov function
$V_{\mathcal Z^*}^{k}$ and derives a compact set containing all
PP-ADMM iterates $\mathbf z^k$ with $k\ge1$. Step 3 bounds the KKT residual
$\|\mathbf R(\mathbf z^{k+1})\|$ by the change between two consecutive
iterates. Step 4 combines these results to prove linear
convergence.

Throughout this proof, the network parameters and problem data
are fixed. Starting from any finite initial point $\mathbf z^0$,
the proximal terms make every local primal block subproblem
coercive and strongly convex in the minimization form. Hence, every primal
update has a unique finite solution. The dual updates are
affine functions of finite primal variables and finite previous
dual variables. Therefore, PP-ADMM generates a well-defined
deterministic sequence.

\subsubsection{Step 1: A uniform KKT error bound holds over compact
sets}
Assumption \ref{assump:kkt_error} gives a local error bound near every saddle point
$\mathbf z^*\in\mathcal Z^*$. The following lemma shows that
this local bound can be made uniform over any compact set
containing at least one saddle point.

\begin{lemma}[Uniform KKT Error Bound]\label{lem:uniform_error_bound}
Under Assumption \ref{assump:kkt_error}, let $\mathcal K\subseteq\mathcal D$ be any compact set
satisfying $\mathcal Z^*\cap\mathcal K\ne\varnothing$. Then the KKT residual
uniformly bounds how far any primal--dual point
$\mathbf z\in\mathcal K$ is from the saddle-point set
$\mathcal Z^*$. Specifically, there exists $\bar\kappa>0$ such that
\begin{equation}
 \inf_{\widehat{\mathbf z}\in\mathcal Z^*}
 \|\mathbf z-\widehat{\mathbf z}\|
 \le
 \bar\kappa\|\mathbf R(\mathbf z)\|,
 \qquad \forall\,\mathbf z\in\mathcal K.
 \label{eq:uniform-kkt-bound}
\end{equation}
\end{lemma}

\begin{IEEEproof} See Appendix \ref{apx:uniform_error_bound}.
\end{IEEEproof}
Thus, the KKT residual $\|\mathbf R(\mathbf z)\|$
controls how close a point $\mathbf z\in\mathcal K$ is to the
saddle-point set $\mathcal Z^*$.

\subsubsection{Step 2: Each PP-ADMM iteration gives sufficient descent
of the Lyapunov function \eqref{eq:lyapunov}}
The following lemma shows that the weighted error
$\|\mathbf z^k-\mathbf z^*\|_{\mathbf G}^{2}$ decreases from
$\mathbf z^k$ to $\mathbf z^{k+1}$ for every saddle point
$\mathbf z^*\in\mathcal Z^*$. Since $V_{\mathcal Z^*}^{k}$ is
the smallest such weighted error over $\mathcal Z^*$, this
result provides sufficient descent of the Lyapunov function.

\begin{lemma}[Lyapunov Sufficient Descent]\label{lem:lyapunov_descent}
Under Assumption \ref{assump:proximal}, there exists $\nu>0$ such that, for every
$\mathbf z^*\in\mathcal Z^*$ and every $k\geq0$,
\begin{equation}
 \|\mathbf z^k-\mathbf z^*\|_{\mathbf G}^{2}
 -
 \|\mathbf z^{k+1}-\mathbf z^*\|_{\mathbf G}^{2}
 \geq
 \nu\|\mathbf z^{k+1}-\mathbf z^k\|_{\mathbf G}^{2}.
 \label{eq:sufficient-descent}
\end{equation}
\end{lemma}

\begin{IEEEproof}
    See Appendix \ref{app:fb-descent}.
\end{IEEEproof} 

We next show that $\mathcal Z^*$ is nonempty. The bounds in
\eqref{eq:P1_xbnd} and \eqref{eq:P1_pbnd} show that $x_f$ and
$p_{ml}$ are bounded. The constraints $y_{ml}^{(n)}\geq 0$ and
\eqref{eq:local_exclusive} imply
$0\leq y_{ml}^{(n)}\leq 1$. Since $p_{ml}$ is
bounded, the SNR equality \eqref{eq:p1-snragg} also bounds
$s_{w,m}$. Since $C_{ml}(p,y)$ is continuous over the bounded
domain of $(p,y)$, the capacity constraints bound $q_{ml}$.
Moreover, at every feasible point, \eqref{eq:P1_rate} and
$r_{ml}^{d}\geq 0$ give
\begin{equation*}
 0\leq r_{ml}^{d}\leq q_{ml}.
\end{equation*}
Therefore, the primal feasible set is closed and bounded, and
hence compact. The objective is continuous over this set, so a
primal optimal solution exists. Assumption~\ref{assump:slater}
guarantees strong duality and that the dual optimum is attained
by at least one finite dual vector. Hence,
$\mathcal Z^*\neq\emptyset$.

Fix any $\mathbf z^*\in\mathcal Z^*$, and write
$\mathbf z^*=
(\boldsymbol\xi^*,\boldsymbol\Lambda^*)$.
Dropping the nonnegative right-hand side of
\eqref{eq:sufficient-descent} and applying the resulting
inequality recursively give
\begin{equation}\label{eq:thm1_z_bound}
\!\! \|\mathbf z^{k+1}-\mathbf z^*\|_{\mathbf G}^2
 \leq
 \|\mathbf z^k-\mathbf z^*\|_{\mathbf G}^2
 \leq
 \|\mathbf z^0-\mathbf z^*\|_{\mathbf G}^2,
 \quad k\geq 0.
\end{equation}
Since $\mathbf G\succ\mathbf 0$, its smallest eigenvalue
$\lambda_{\min}(\mathbf G)$ is positive. Therefore,
\begin{equation}
 \|\mathbf z^k-\mathbf z^*\|^2
 \leq
 \frac{1}{\lambda_{\min}(\mathbf G)}
 \|\mathbf z^0-\mathbf z^*\|_{\mathbf G}^2,
 \qquad k\geq 0.
\end{equation}
Because $\mathbf z^k$ contains all primal and dual variables,
this inequality uniformly bounds the complete primal--dual
sequence. In particular, since the dual block of $\mathbf G$
is $\frac{1}{\rho\tau}\mathbf I$,
\begin{equation}
 \frac{1}{\rho\tau}
 \|\boldsymbol\Lambda^k-\boldsymbol\Lambda^*\|^2
 \leq
 \|\mathbf z^k-\mathbf z^*\|_{\mathbf G}^2
 \leq
 \|\mathbf z^0-\mathbf z^*\|_{\mathbf G}^2.
\end{equation}
Hence,
\begin{equation}
 \|\boldsymbol\Lambda^k\|
 \leq
 \|\boldsymbol\Lambda^*\|
 +
 \sqrt{\rho\tau}\,
 \|\mathbf z^0-\mathbf z^*\|_{\mathbf G}
 <\infty,
 \qquad k\geq 0.
\end{equation}
Thus, the complete sequence of dual iterates is uniformly
bounded.

Define
\begin{equation}
\mathcal K
\triangleq
\left\{
\mathbf z\in\mathcal D:
\|\mathbf z-\mathbf z^*\|_G^2
\le
\|\mathbf z^0-\mathbf z^*\|_G^2
\right\}.
\label{eq:thm1_K_def}
\end{equation}
Since $\mathcal D$ is closed and $G\succ0$,
$\mathcal K$ is compact. It contains the fixed saddle point $\mathbf z^*$
and every PP-ADMM iterate $\mathbf z^k$ with $k\ge1$.

\subsubsection{Step 3: One PP-ADMM iteration bounds the KKT residual}
The iterate change controlled in Step 2 provides the following
bound on the KKT residual
$\|\mathbf R(\mathbf z^{k+1})\|$.

\begin{lemma}[One-Step KKT Residual Bound]\label{lem:residual_bound}
Under Assumption \ref{assump:proximal}, there exists $C_R>0$ such that
\begin{equation}
 \|\mathbf R(\mathbf z^{k+1})\|
 \leq
 C_R\|\mathbf z^{k+1}-\mathbf z^k\|_{\mathbf G},
 \qquad k\geq0.
 \label{eq:one-step-residual}
\end{equation}
\end{lemma}

\begin{IEEEproof}
    See Appendix \ref{apx:residual_bound}.
\end{IEEEproof}

Therefore, when the PP-ADMM iterate changes approach zero, the KKT
residual also approaches zero.

\subsubsection{Step 4: The three bounds imply linear convergence}
We first show that a closest saddle point exists for every
iterate. The residual $\mathbf R (\mathbf z)$ is continuous on the closed set
$\mathcal D$. Therefore,
\[
\mathcal Z^*
=
\{\mathbf z\in\mathcal D:\mathbf R(\mathbf z)=0\}
\]
is closed. Moreover, since $\mathbf G\succ0$,
$\|\mathbf z^k-\mathbf z^*\|_{\mathbf G}^{2}$ is continuous and
grows without bound as $\|\mathbf z^*\|\rightarrow\infty$.
Therefore, the infimum defining $V_{\mathcal Z^*}^{k}$ is
attained. For every $k$, choose
$\mathbf z^{k,*}\in\mathcal Z^*$ such that
\begin{equation}
 V_{\mathcal Z^*}^{k}
 =
 \|\mathbf z^k-\mathbf z^{k,*}\|_{\mathbf G}^{2}.
 \label{eq:closest-saddle}
\end{equation}

Apply Lemma \ref{lem:lyapunov_descent} with
$\mathbf z^*=\mathbf z^{k,*}$. Since
$\mathbf z^{k,*}\in\mathcal Z^*$ is also a candidate in the
definition of $V_{\mathcal Z^*}^{k+1}$, we have
\begin{align}
 V_{\mathcal Z^*}^{k}-V_{\mathcal Z^*}^{k+1}
 &\geq
 \|\mathbf z^k-\mathbf z^{k,*}\|_{\mathbf G}^{2}
 -
 \|\mathbf z^{k+1}-\mathbf z^{k,*}\|_{\mathbf G}^{2}
 \nonumber\\
 &\geq
 \nu\|\mathbf z^{k+1}-\mathbf z^k\|_{\mathbf G}^{2}.
 \label{eq:V-descent}
\end{align}
Thus, the decrease in the Lyapunov function bounds the change
between $\mathbf z^k$ and $\mathbf z^{k+1}$.

We next use this iterate change to bound the remaining error.
By \eqref{eq:thm1_z_bound} and the definition of $\mathcal K$ in \eqref{eq:thm1_K_def},
$\mathbf z^{k+1}\in\mathcal K$.
Using
$\|\mathbf u\|_{\mathbf G}^{2}
\leq\lambda_{\max}(\mathbf G)\|\mathbf u\|^{2}$,
the uniform KKT error bound \eqref{eq:uniform-kkt-bound}, and
the one-step residual bound \eqref{eq:one-step-residual}, we
obtain
\begin{align}
 V_{\mathcal Z^*}^{k+1}
 &\leq
 \lambda_{\max}(\mathbf G)
 \left(
 \inf_{\mathbf z^*\in\mathcal Z^*}
 \|\mathbf z^{k+1}-\mathbf z^*\|
 \right)^2
 \nonumber\\
 &\leq
 \lambda_{\max}(\mathbf G)\bar\kappa^2
 \|\mathbf R(\mathbf z^{k+1})\|^2
 \nonumber\\
 &\leq
 \lambda_{\max}(\mathbf G)\bar\kappa^2 C_R^2
 \|\mathbf z^{k+1}-\mathbf z^k\|_{\mathbf G}^{2}.
 \label{eq:V-residual-bound}
\end{align}

Define
\begin{equation}
 C_V
 \triangleq
 \lambda_{\max}(\mathbf G)\bar\kappa^2 C_R^2>0.
 \label{eq:CV}
\end{equation}
Then \eqref{eq:V-residual-bound} becomes
\begin{equation}
 V_{\mathcal Z^*}^{k+1}
 \leq
 C_V\|\mathbf z^{k+1}-\mathbf z^k\|_{\mathbf G}^{2}.
 \label{eq:V-step-bound}
\end{equation}

Combining \eqref{eq:V-descent} and
\eqref{eq:V-step-bound} gives
\begin{equation}
 V_{\mathcal Z^*}^{k}-V_{\mathcal Z^*}^{k+1}
 \geq
 \nu\|\mathbf z^{k+1}-\mathbf z^k\|_{\mathbf G}^{2}
 \geq
 \frac{\nu}{C_V}V_{\mathcal Z^*}^{k+1}.
 \label{eq:combined-contraction}
\end{equation}
Hence,
\begin{equation}
 V_{\mathcal Z^*}^{k}
 \geq
 \left(1+\frac{\nu}{C_V}\right)
 V_{\mathcal Z^*}^{k+1}.
 \label{eq:contraction-rearranged}
\end{equation}
Equivalently,
\begin{equation}
V_{\mathcal Z^*}^{k+1}
\le
\beta V_{\mathcal Z^*}^{k},
\qquad
\beta\triangleq\frac{C_V}{C_V+\nu}\in(0,1).
\label{eq:q-contraction}
\end{equation}
This proves
\eqref{eq:qlinear}. Applying \eqref{eq:q-contraction} repeatedly
yields
\begin{equation}
 V_{\mathcal Z^*}^{k}
 \leq
 \beta^kV_{\mathcal Z^*}^{0}.
 \label{eq:Q-linear}
\end{equation}

It remains to prove convergence of the complete sequence
$\{\mathbf z^k\}_{k\geq0}$ to a single saddle point. Since
$V_{\mathcal Z^*}^{k+1}\geq0$, \eqref{eq:V-descent} and
\eqref{eq:Q-linear} give
\begin{equation}
 \nu\|\mathbf z^{k+1}-\mathbf z^k\|_{\mathbf G}^{2}
 \leq
 V_{\mathcal Z^*}^{k}-V_{\mathcal Z^*}^{k+1}
 \leq
 V_{\mathcal Z^*}^{k}
 \leq
 \beta^kV_{\mathcal Z^*}^{0}.
 \label{eq:step-decay-squared}
\end{equation}
Therefore,
\begin{equation}
 \|\mathbf z^{k+1}-\mathbf z^k\|_{\mathbf G}
 \leq
 \sqrt{\frac{V_{\mathcal Z^*}^{0}}{\nu}}\,\beta^{k/2}.
 \label{eq:step-decay}
\end{equation}

For any integers $r>s$, the triangle inequality and
\eqref{eq:step-decay} give
\begin{align}
 \|\mathbf z^r-\mathbf z^s\|_{\mathbf G}
 &\leq
 \sum_{j=s}^{r-1}
 \|\mathbf z^{j+1}-\mathbf z^j\|_{\mathbf G}
 \nonumber\\
 &\leq
 \sqrt{\frac{V_{\mathcal Z^*}^{0}}{\nu}}
 \sum_{j=s}^{\infty}\beta^{j/2}
 \nonumber\\
 &=
 \sqrt{\frac{V_{\mathcal Z^*}^{0}}{\nu}}\,
 \frac{\beta^{s/2}}{1-\sqrt \beta}.
 \label{eq:Cauchy-bound}
\end{align}
Because $0<\beta<1$, the last expression approaches zero as
$s\rightarrow\infty$. Thus, any two sufficiently late iterates
are arbitrarily close. Hence, $\{\mathbf z^k\}_{k\geq0}$ is a
Cauchy sequence and converges to some point $\mathbf z^\infty$.

Lemma \ref{lem:residual_bound} and \eqref{eq:step-decay} further
give
\begin{equation}
 \|\mathbf R(\mathbf z^{k+1})\|
 \leq
 C_R\sqrt{\frac{V_{\mathcal Z^*}^{0}}{\nu}}\,
 \beta^{k/2}
 \longrightarrow0.
 \label{eq:residual-convergence}
\end{equation}
Since $\mathbf z^{k+1}\in\mathcal D$ and $\mathcal D$
is closed, $\mathbf z^\infty\in\mathcal D$.
Continuity of $\mathbf R$ on $\mathcal D$ and \eqref{eq:residual-convergence}
then give $\mathbf R(\mathbf z^\infty)=0$.
Hence, $\mathbf z^\infty\in\mathcal Z^*$.

Finally, letting $r\rightarrow\infty$ and setting $s=k$ in
\eqref{eq:Cauchy-bound} give
\begin{equation}
 \|\mathbf z^k-\mathbf z^\infty\|_{\mathbf G}
 \leq
 \sqrt{\frac{V_{\mathcal Z^*}^{0}}{\nu}}\,
 \frac{\beta^{k/2}}{1-\sqrt \beta}.
 \label{eq:sequence-bound}
\end{equation}
Using
$\|\mathbf u\|^2\leq
\lambda_{\min}(\mathbf G)^{-1}\|\mathbf u\|_{\mathbf G}^2$
and squaring \eqref{eq:sequence-bound} yield
\begin{equation}
 \!\!\!\|\mathbf z^k-\mathbf z^\infty\|^{2}
 \leq
 C\beta^kV_{\mathcal Z^*}^{0},
 \quad
 C\triangleq
 \frac{1}{
 \lambda_{\min}(\mathbf G)\nu(1-\sqrt \beta)^2}.
 \label{eq:R-linear}
\end{equation}
This proves \eqref{eq:rlinear}. Therefore,
$V_{\mathcal Z^*}^{k}$ converges Q-linearly, and the complete
primal--dual sequence $\{\mathbf z^k\}_{k\geq0}$ converges
R-linearly to a saddle point of the problem under
consideration. By Proposition~\ref{prop:local_equivalence}, the
primal component of this saddle point is an optimal solution of
Problem \eqref{prob:P1} or \eqref{prob:unified}, with the same
optimal objective value.

\section{Proof of Lemma~\ref{lem:uniform_error_bound}}\label{apx:uniform_error_bound}

\begin{IEEEproof}
The proof proceeds in three steps. First, we apply the local
KKT error bound \eqref{eq:kkt_error_bound} to points
$\mathbf z\in\mathcal K$ near the saddle points
$\mathbf z^*\in\mathcal Z^*\cap\mathcal K$. Second, we bound
the remaining points in $\mathcal K$ outside these
neighborhoods. Finally, we combine the two bounds to obtain one
constant $\bar\kappa$ valid for every
$\mathbf z\in\mathcal K$.

\subsubsection{Step 1: Apply the Local Bound Near the Saddle Points}
The residual mapping $\mathbf R(\mathbf z)$ is continuous on
$\mathcal D$ because Euclidean projection onto a nonempty
closed convex set is continuous and the remaining
functions in $\mathbf R(\mathbf z)$ are continuous on this domain.
Moreover,
\[
\mathcal Z^*
=
\{\mathbf z\in\mathcal D:\mathbf R(\mathbf z)=0\}.
\]
Since $\mathcal D$ is closed, $\mathcal Z^*$ is closed.
By the conditions of Lemma 2, $\mathcal K$ is compact
and $\mathcal Z^*\cap\mathcal K$ is nonempty.
Hence, $\mathcal Z^*\cap\mathcal K$ is compact.

For every $\mathbf z^*\in\mathcal Z^*\cap\mathcal K$,
Assumption \ref{assump:kkt_error} gives constants
$\epsilon_{\mathbf z^*}>0$ and
$\kappa_{\mathbf z^*}>0$ for which
\eqref{eq:kkt_error_bound} holds. The corresponding open
neighborhoods cover $\mathcal Z^*\cap\mathcal K$. Since this
set is compact, finitely many neighborhoods are sufficient.
Thus, there exist saddle points
$\mathbf z_1^*,\ldots,\mathbf z_J^*
\in\mathcal Z^*\cap\mathcal K$ such that
\begin{equation}
 \mathcal Z^*\cap\mathcal K
 \subseteq
 \bigcup_{i=1}^{J}
 \left\{
 \mathbf z:
 \|\mathbf z-\mathbf z_i^*\|
 <
 \epsilon_{\mathbf z_i^*}
 \right\}.
 \label{eq:finite-cover}
\end{equation}

Define
\begin{equation}
 \mathcal U
 \triangleq
 \bigcup_{i=1}^{J}
 \left\{
 \mathbf z:
 \|\mathbf z-\mathbf z_i^*\|
 <
 \epsilon_{\mathbf z_i^*}
 \right\},
 \qquad
 \kappa_1
 \triangleq
 \max_{1\leq i\leq J}\kappa_{\mathbf z_i^*}.
 \label{eq:neighborhood-union}
\end{equation}
Every $\mathbf z\in\mathcal K\cap\mathcal U$ belongs to at
least one of these neighborhoods. Applying
\eqref{eq:kkt_error_bound} in that neighborhood gives
\begin{equation}
 \inf_{\mathbf z^*\in\mathcal Z^*}
 \|\mathbf z-\mathbf z^*\|
 \leq
 \kappa_1\|\mathbf R(\mathbf z)\|,
 \qquad
 \mathbf z\in\mathcal K\cap\mathcal U.
 \label{eq:near-saddle-bound}
\end{equation}

\subsubsection{Step 2: Establish the Error Bound
Outside the Neighborhoods in \eqref{eq:neighborhood-union}}
Define
\[
 \mathcal K_0\triangleq\mathcal K\setminus\mathcal U.
\]
If $\mathcal K_0=\emptyset$, set
$\bar\kappa\triangleq\kappa_1$. Then
\eqref{eq:near-saddle-bound} proves
\eqref{eq:uniform-kkt-bound}, so it remains to consider
$\mathcal K_0\neq\emptyset$.

The set $\mathcal K_0$ is compact because it is a closed subset
of the compact set $\mathcal K$. It contains no saddle point
because $\mathcal U$ covers $\mathcal Z^*\cap\mathcal K$.
Therefore,
$\mathbf R(\mathbf z)\neq\mathbf 0$ for every
$\mathbf z\in\mathcal K_0$.

Since $\mathbf R(\mathbf z)$ is continuous, its norm attains a
minimum over $\mathcal K_0$. Define
\begin{equation}
 R_{\mathrm{min}}
 \triangleq
 \min_{\mathbf z\in\mathcal K_0}
 \|\mathbf R(\mathbf z)\|.
 \label{eq:minimum-residual}
\end{equation}
We have $R_{\mathrm{min}}>0$. Otherwise, the minimum would be
attained at some $\mathbf z\in\mathcal K_0$ satisfying
$\mathbf R(\mathbf z)=\mathbf 0$. This would imply
$\mathbf z\in\mathcal Z^*\cap\mathcal K\subseteq\mathcal U$,
which contradicts $\mathbf z\in\mathcal K_0$.

Choose any reference saddle point
$\mathbf z^{\mathrm{ref}}\in\mathcal Z^*\cap\mathcal K$.
Since $\mathcal K$ is compact, define
\begin{equation}
 d_{\mathrm{ref}}
 \triangleq
 \max_{\mathbf z\in\mathcal K}
 \|\mathbf z-\mathbf z^{\mathrm{ref}}\|
 <\infty.
 \label{eq:maximum-reference-distance}
\end{equation}
For every $\mathbf z\in\mathcal K_0$, the reference point
$\mathbf z^{\mathrm{ref}}$ is a candidate in the infimum over
$\mathcal Z^*$. Hence,
\begin{align}
 \inf_{\mathbf z^*\in\mathcal Z^*}
 \|\mathbf z-\mathbf z^*\|
 &\leq
 \|\mathbf z-\mathbf z^{\mathrm{ref}}\|
 \nonumber\\
 &\leq d_{\mathrm{ref}}
 \nonumber\\
 &\leq
 \frac{d_{\mathrm{ref}}}{R_{\mathrm{min}}}
 \|\mathbf R(\mathbf z)\|,
 \qquad
 \mathbf z\in\mathcal K_0,
 \label{eq:outside-neighborhood-bound}
\end{align}
where the last inequality follows from
\eqref{eq:minimum-residual}.

\subsubsection{Step 3: Combine the Two Bounds}
Define
\begin{equation}
\bar\kappa
\triangleq
\max\left\{
\kappa_1,\frac{d_{\mathrm{ref}}}{R_{\mathrm{min}}}
\right\}
>0.
\label{eq:uniform-error-constant}
\end{equation}
Equation \eqref{eq:near-saddle-bound} applies to
$\mathbf z\in\mathcal K\cap\mathcal U$, while
\eqref{eq:outside-neighborhood-bound} applies to
$\mathbf z\in\mathcal K_0$. These two sets cover $\mathcal K$.
Therefore,
\begin{equation}
 \inf_{\mathbf z^*\in\mathcal Z^*}
 \|\mathbf z-\mathbf z^*\|
 \leq
 \bar\kappa\|\mathbf R(\mathbf z)\|,
 \qquad
 \forall\,\mathbf z\in\mathcal K.
 \label{eq:uniform-kkt-bound-proof}
\end{equation}
This proves \eqref{eq:uniform-kkt-bound}.
\end{IEEEproof}

\section{Proof of Lemma~\ref{lem:lyapunov_descent}}\label{app:fb-descent}

\begin{IEEEproof}
The proof proceeds in four steps.
Step 1 compares the optimality conditions of the primal
updates~\eqref{eq:fb-primal-update} with the KKT conditions
of Problem~\eqref{eq:fb-problem}.
Step 2 uses this comparison and the dual
update~\eqref{eq:fb-dual-update} to bound the decrease
of the Lyapunov function~\eqref{eq:lyapunov}.
Step 3 bounds the products between primal and dual
iterate changes in~\eqref{eq:fb-descent-cross} and applies
the parameter bound~\eqref{eq:fb-alpha}.
Step 4 compares the resulting bound with the matrix
$G$ in~\eqref{eq:fb-G} to obtain the positive descent
constant in Lemma~\ref{lem:lyapunov_descent}.

Fix an iteration $k\ge0$ and a saddle point
$\mathbf z^*
=(\boldsymbol\xi^*,\boldsymbol\Lambda^*)
\in\mathcal Z^*$ of Problem~\eqref{eq:fb-problem}.
Define the primal and dual iterate changes by
\begin{equation}
\begin{aligned}
&\mathbf d_i
\triangleq
\boldsymbol\xi_i^{k+1}-\boldsymbol\xi_i^k,
\qquad i=1,\ldots,4,\\
&\mathbf d_\Lambda
\triangleq
\boldsymbol\Lambda^{k+1}-\boldsymbol\Lambda^k,\\
&\mathbf d
\triangleq
\begin{bmatrix}
\mathbf d_1\\
\mathbf d_2\\
\mathbf d_3\\
\mathbf d_4
\end{bmatrix}.
\end{aligned}
\label{eq:fb-iterate-changes}
\end{equation}

\subsubsection{Step 1: Compare the Primal Updates
\eqref{eq:fb-primal-update} with the KKT Conditions
of Problem~\eqref{eq:fb-problem}}

The optimality condition of the $i$th primal
update~\eqref{eq:fb-primal-update} is
\begin{equation}
\begin{aligned}
0\in{}&
\partial\Phi_i(\boldsymbol\xi_i^{k+1})
+Q_i^{\mathsf T}\boldsymbol\Lambda^k\\
&+\rho Q_i^{\mathsf T}
\left(
Q\boldsymbol\xi^{k+1}
-\sum_{j\ne i}Q_j\mathbf d_j
\right)
+P_i\mathbf d_i.
\end{aligned}
\label{eq:fb-update-optimality}
\end{equation}
Here, the definition of $\mathbf d_j$
in~\eqref{eq:fb-iterate-changes} gives
\[
Q\boldsymbol\xi^{k+1}
-\sum_{j\ne i}Q_j\mathbf d_j
=
Q_i\boldsymbol\xi_i^{k+1}
+\sum_{j\ne i}Q_j\boldsymbol\xi_j^k.
\]

At the saddle point $\mathbf z^*$, the KKT
conditions of Problem~\eqref{eq:fb-problem} are
\begin{equation}
\begin{aligned}
-Q_i^{\mathsf T}\boldsymbol\Lambda^*
&\in\partial\Phi_i(\boldsymbol\xi_i^*),
\qquad i=1,\ldots,4,\\
Q\boldsymbol\xi^*&=0.
\end{aligned}
\label{eq:fb-saddle-kkt}
\end{equation}
For convenience, define
\begin{equation}
\widehat{\boldsymbol\Lambda}
\triangleq
\boldsymbol\Lambda^k+\rho Q\boldsymbol\xi^{k+1}.
\label{eq:fb-auxiliary-dual}
\end{equation}
Using~\eqref{eq:fb-auxiliary-dual} and
$H_i=P_i+\rho Q_i^{\mathsf T}Q_i$
from~\eqref{eq:fb-H}, we can rewrite
\eqref{eq:fb-update-optimality} as
\begin{equation}
-Q_i^{\mathsf T}\widehat{\boldsymbol\Lambda}
+\rho Q_i^{\mathsf T}Q\mathbf d-H_i\mathbf d_i
\in\partial\Phi_i(\boldsymbol\xi_i^{k+1}).
\label{eq:fb-update-subgradient}
\end{equation}

Each $\Phi_i$ is convex, so its subdifferential is
monotone. Comparing the subgradients
in~\eqref{eq:fb-update-subgradient}
and~\eqref{eq:fb-saddle-kkt}, and summing over
$i=1,\ldots,4$, gives
\begin{equation}
\begin{aligned}
&\sum_{i=1}^{4}
\left\langle
\boldsymbol\xi_i^{k+1}-\boldsymbol\xi_i^*,
H_i\mathbf d_i
\right\rangle\\
&\quad\le
-\left\langle
Q\boldsymbol\xi^{k+1},
\widehat{\boldsymbol\Lambda}-\boldsymbol\Lambda^*
\right\rangle
+\rho\left\langle
Q\boldsymbol\xi^{k+1},Q\mathbf d
\right\rangle.
\end{aligned}
\label{eq:fb-monotonicity}
\end{equation}
In obtaining~\eqref{eq:fb-monotonicity}, we used
the feasibility condition
$Q\boldsymbol\xi^*=0$
in~\eqref{eq:fb-saddle-kkt}.

\subsubsection{Step 2: Bound the Decrease of the
Lyapunov Function~\eqref{eq:lyapunov}}

The dual update~\eqref{eq:fb-dual-update} and
the definition~\eqref{eq:fb-auxiliary-dual} imply
\begin{equation}
\begin{aligned}
Q\boldsymbol\xi^{k+1}
&=\frac{\mathbf d_\Lambda}{\rho\tau},\\
\widehat{\boldsymbol\Lambda}
-\boldsymbol\Lambda^{k+1}
&=\frac{1-\tau}{\tau}\mathbf d_\Lambda.
\end{aligned}
\label{eq:fb-dual-identities}
\end{equation}

For any symmetric matrix $M$ and vectors of
the corresponding dimension,
\begin{equation}
\begin{aligned}
\|\mathbf a-\mathbf c\|_M^2
-\|\mathbf b-\mathbf c\|_M^2
={}&
2\left\langle
\mathbf a-\mathbf b,M(\mathbf b-\mathbf c)
\right\rangle\\
&+\|\mathbf a-\mathbf b\|_M^2.
\end{aligned}
\label{eq:fb-norm-identity}
\end{equation}
Apply~\eqref{eq:fb-norm-identity} to the four primal
terms and the dual term in the Lyapunov
expression~\eqref{eq:fb-expanded-lyapunov}.
Using the iterate changes
in~\eqref{eq:fb-iterate-changes}, we obtain
\begin{align}
&V(\mathbf z^k;\mathbf z^*)
-V(\mathbf z^{k+1};\mathbf z^*)
=
\sum_{i=1}^{4}\|\mathbf d_i\|_{H_i}^2
+\frac{1}{\rho\tau}\|\mathbf d_\Lambda\|^2
\notag\\
&-2\sum_{i=1}^{4}
\langle
\boldsymbol\xi_i^{k+1}-\boldsymbol\xi_i^*,
H_i\mathbf d_i
\rangle
-\frac{2}{\rho\tau}
\langle
\boldsymbol\Lambda^{k+1}-\boldsymbol\Lambda^*,
\mathbf d_\Lambda
\rangle.
\label{eq:fb-lyapunov-difference}
\end{align}

Substitute the bound~\eqref{eq:fb-monotonicity}
into~\eqref{eq:fb-lyapunov-difference}, and then use
the dual identities~\eqref{eq:fb-dual-identities}.
The inner products involving
$\boldsymbol\Lambda^{k+1}-\boldsymbol\Lambda^*$
cancel. Hence,
\begin{equation}
\begin{aligned}
&V(\mathbf z^k;\mathbf z^*)
-V(\mathbf z^{k+1};\mathbf z^*)\\
&\ge
\sum_{i=1}^{4}\|\mathbf d_i\|_{H_i}^2
+\frac{2-\tau}{\rho\tau^2}\|\mathbf d_\Lambda\|^2
-\frac{2}{\tau}\sum_{i=1}^{4}
\left\langle\mathbf d_\Lambda,Q_i\mathbf d_i\right\rangle.
\end{aligned}
\label{eq:fb-descent-cross}
\end{equation}
The first two terms on the right are nonnegative.
The last line contains inner products between the
dual change $\mathbf d_\Lambda$ and the primal
changes $Q_i\mathbf d_i$. These inner products can
have either sign, so we bound them next.

\subsubsection{Step 3: Bound the Primal--Dual Inner
Products in~\eqref{eq:fb-descent-cross}}

Let $\delta$ be the constant in Assumption 2 and set
\begin{equation}
\epsilon_i
\triangleq\frac{2-\tau-\delta}{4},
\qquad i=1,\ldots,4.
\label{eq:fb-young-weights}
\end{equation}
Since $0<\delta<2-\tau$ and $0<\tau<2$,
we have $0<\epsilon_i<1$ and
\[
\sum_{i=1}^{4}\epsilon_i=2-\tau-\delta.
\]

Apply Young's inequality to each inner product
in~\eqref{eq:fb-descent-cross}. This gives
\begin{equation}
\begin{aligned}
-\frac{2}{\tau}
\left\langle\mathbf d_\Lambda,Q_i\mathbf d_i\right\rangle
\ge{}&
-\frac{\epsilon_i}{\rho\tau^2}
 \|\mathbf d_\Lambda\|^2-\frac{\rho}{\epsilon_i}\|Q_i\mathbf d_i\|^2.
\end{aligned}
\label{eq:fb-young}
\end{equation}
Define the remaining primal weights by
\begin{equation}
\begin{aligned}
K_i
&\triangleq
H_i-\frac{\rho}{\epsilon_i}Q_i^{\mathsf T}Q_i\\
&=
\alpha I+\rho D_i
-\frac{4\rho}{2-\tau-\delta}Q_i^{\mathsf T}Q_i.
\end{aligned}
\label{eq:fb-K}
\end{equation}
The second equality follows from~\eqref{eq:fb-H}
and~\eqref{eq:fb-young-weights}.
The parameter bound~(30) in
Assumption 2 implies $K_i\succ0$.

The same bound also verifies the positivity of the
proximal matrices used in~\eqref{eq:fb-primal-update}.
Indeed, since $\epsilon_i<1$,
definitions~\eqref{eq:fb-P} and~\eqref{eq:fb-K} give
\begin{equation}
P_i
=
K_i+\rho(\epsilon_i^{-1}-1)Q_i^{\mathsf T}Q_i
\succ0.
\end{equation}

Substituting~\eqref{eq:fb-young}
into~\eqref{eq:fb-descent-cross} and collecting
the primal terms according to~\eqref{eq:fb-K},
we obtain
\begin{equation}
\begin{aligned}
\!\!\!\!\!V(\mathbf z^k;\mathbf z^*)
-\!V(\mathbf z^{k+1};\mathbf z^*)\ge\!
\sum_{i=1}^{4}\|\mathbf d_i\|_{K_i}^2
+\!\frac{\delta}{\rho\tau^2}\|\mathbf d_\Lambda\|^2.
\end{aligned}
\label{eq:fb-positive-descent}
\end{equation}
Here, the coefficient of the dual term is
$\delta/(\rho\tau^2)$ because
$2-\tau-\sum_{i=1}^{4}\epsilon_i=\delta$.
All matrix weights on the right of
\eqref{eq:fb-positive-descent} are positive definite.

\subsubsection{Step 4: Obtain the $G$-Norm Descent
Bound in Lemma~\ref{lem:lyapunov_descent}}

Collect the weights in~\eqref{eq:fb-positive-descent}
into the matrix
\begin{equation}
T\triangleq
\begin{bmatrix}
K_1 & 0   & 0   & 0   & 0\\
0   & K_2 & 0   & 0   & 0\\
0   & 0   & K_3 & 0   & 0\\
0   & 0   & 0   & K_4 & 0\\
0   & 0   & 0   & 0   & \dfrac{\delta}{\rho\tau^2}I
\end{bmatrix}.
\label{eq:fb-T}
\end{equation}
By~\eqref{eq:fb-state} and~\eqref{eq:fb-iterate-changes},
\begin{align}
\mathbf z^{k+1}-\mathbf z^k
=
\begin{bmatrix}
\mathbf d\\
\mathbf d_\Lambda
\end{bmatrix}
=
\begin{bmatrix}
\mathbf d_1\\
\mathbf d_2\\
\mathbf d_3\\
\mathbf d_4\\
\mathbf d_\Lambda
\end{bmatrix}.
\end{align}
Thus, the right-hand side of
\eqref{eq:fb-positive-descent} equals
$\|\mathbf z^{k+1}-\mathbf z^k\|_T^2$.

The matrix $T$ in~\eqref{eq:fb-T} and the Lyapunov
matrix $G$ in~\eqref{eq:fb-G} are positive definite.
Define
\begin{equation}
\nu
\triangleq
\lambda_{\min}(G^{-1/2}TG^{-1/2})>0,
\label{eq:fb-descent-constant}
\end{equation}
where $\lambda_{\min}(M)$ denotes the smallest
eigenvalue of a symmetric matrix $M$.
It follows that $T\succeq\nu G$. Therefore,
\eqref{eq:fb-positive-descent} gives
\begin{equation}
\begin{aligned}
V(\mathbf z^k;\mathbf z^*)
-V(\mathbf z^{k+1};\mathbf z^*)&\ge
\|\mathbf z^{k+1}-\mathbf z^k\|_T^2\\&\ge
\nu\|\mathbf z^{k+1}-\mathbf z^k\|_G^2.
\end{aligned}
\label{eq:fb-final-descent}
\end{equation}
The matrices $T$ and $G$ depend only on the fixed
network data and algorithm parameters.
Hence, $\nu$ is independent of $k$ and the chosen
saddle point $\mathbf z^*$.
This proves Lemma~\ref{lem:lyapunov_descent}.
\end{IEEEproof}

\section{Proof of Lemma~\ref{lem:residual_bound}}\label{apx:residual_bound}

\begin{IEEEproof}
The proof proceeds in three steps.
Step 1 rewrites the optimality condition of the primal
update~\eqref{eq:fb-primal-update} using the gradient
of \eqref{eq:auglag} at $\mathbf z^{k+1}$ and the changes between
iterations $k$ and $k+1$.
Step 2 uses the primal optimality condition
in ~\eqref{eq:fb-perturbed-stationarity} 
to bound
the four projection differences in the KKT
residual~\eqref{eq:fb-kkt-residual}.
Step 3 uses the dual update~\eqref{eq:fb-dual-update}
to bound the residuals of equalities
\eqref{eq:P1_flow}, \eqref{eq:P1_rate}, \eqref{eq:p1-snragg}, and \eqref{eq:consensus}, and combines these
bounds to prove Lemma~\ref{lem:residual_bound}.

Fix $k\ge0$, and let $\mathbf z^{k+1}$ contain the
primal variables generated by~\eqref{eq:fb-primal-update}
and the dual variables generated by~\eqref{eq:fb-dual-update}.
We use the normal cone to express the local constraints
$\boldsymbol\xi_i\in\mathcal C_i$.
For a nonempty closed convex set $\mathcal C$ and
$\mathbf u\in\mathcal C$, define
\begin{equation}
N_{\mathcal C}(\mathbf u)
\triangleq
\left\{
\mathbf v:
\langle\mathbf v,\mathbf w-\mathbf u\rangle\le0
\text{ for all }\mathbf w\in\mathcal C
\right\}.
\label{eq:fb-normal-cone}
\end{equation}

\subsubsection{Step 1: Rewrite the Primal Optimality
Condition~\eqref{eq:fb-update-optimality}
Using the Gradient of \eqref{eq:auglag} at $\mathbf z^{k+1}$}

Recall that $\Phi_i=f_i+\iota_{\mathcal C_i}$.
Writing the local constraint through its normal cone,
the optimality condition~\eqref{eq:fb-update-optimality}
becomes
\begin{equation}
\begin{aligned}
0\in{}&
\nabla f_i(\boldsymbol\xi_i^{k+1})
+N_{\mathcal C_i}(\boldsymbol\xi_i^{k+1})
+Q_i^{\mathsf T}\boldsymbol\Lambda^k\\
&+\rho Q_i^{\mathsf T}
\left(
Q\boldsymbol\xi^{k+1}
-\sum_{j\ne i}Q_j\mathbf d_j
\right)
+P_i\mathbf d_i.
\end{aligned}
\label{eq:fb-residual-optimality}
\end{equation}

The function $f_i$ is the negative of the corresponding
utility term in \eqref{eq:auglag}. Hence, the gradient of \eqref{eq:auglag}
with respect to block $\boldsymbol\xi_i$, evaluated
at the complete iterate $\mathbf z^{k+1}$, is
\begin{equation}
\begin{aligned}
\nabla_{\boldsymbol\xi_i}\mathcal L_\rho(\mathbf z^{k+1})
={}&\!
-\nabla f_i(\boldsymbol\xi_i^{k+1})
-\!Q_i^{\mathsf T}\boldsymbol\Lambda^{k+1}-\!\rho Q_i^{\mathsf T}Q\boldsymbol\xi^{k+1}.
\end{aligned}
\label{eq:fb-lagrangian-gradient}
\end{equation}
In particular, all primal and dual variables in
\eqref{eq:fb-lagrangian-gradient} are evaluated
at iteration $k+1$.

Substitute~\eqref{eq:fb-lagrangian-gradient}
into~\eqref{eq:fb-residual-optimality}.
Using
$\boldsymbol\Lambda^{k+1}-\boldsymbol\Lambda^k
=\mathbf d_\Lambda$
and $H_i=P_i+\rho Q_i^{\mathsf T}Q_i$
from~\eqref{eq:fb-H}, collect the terms involving
iterate changes into
\begin{equation}
\mathbf e_i
\triangleq
-H_i\mathbf d_i
+\rho Q_i^{\mathsf T}Q\mathbf d
+Q_i^{\mathsf T}\mathbf d_\Lambda.
\label{eq:fb-perturbation}
\end{equation}
The optimality condition then gives
\begin{equation}
\mathbf e_i
\in
-\nabla_{\boldsymbol\xi_i}\mathcal L_\rho(\mathbf z^{k+1})
+N_{\mathcal C_i}(\boldsymbol\xi_i^{k+1}).
\label{eq:fb-perturbed-stationarity}
\end{equation}
The vector $\mathbf e_i$ is a fixed linear function
of the primal and dual iterate changes.
We next use~\eqref{eq:fb-perturbed-stationarity}
to bound the four projection differences
in~\eqref{eq:fb-kkt-residual}.

\subsubsection{Step 2: Bound the Four Projection
Differences in the KKT Residual~\eqref{eq:fb-kkt-residual}}
We bound the first four block rows of
\eqref{eq:fb-kkt-residual} at $\mathbf z^{k+1}$.
For $\mathbf u\in\mathcal C$, the condition
$\mathbf v\in N_{\mathcal C}(\mathbf u)$ holds
if and only if
\begin{equation}
\mathbf u=\Pi_{\mathcal C}(\mathbf u+\mathbf v).
\label{eq:fb-normal-projection}
\end{equation}
By~\eqref{eq:fb-perturbed-stationarity},
\[
\nabla_{\boldsymbol\xi_i}\mathcal L_\rho(\mathbf z^{k+1})
+\mathbf e_i
\in N_{\mathcal C_i}(\boldsymbol\xi_i^{k+1}).
\]
Applying~\eqref{eq:fb-normal-projection} therefore gives
\begin{equation}
\boldsymbol\xi_i^{k+1}
=
\Pi_{\mathcal C_i}\!\left(
\boldsymbol\xi_i^{k+1}
+\nabla_{\boldsymbol\xi_i}\mathcal L_\rho(\mathbf z^{k+1})
+\mathbf e_i
\right).
\label{eq:fb-perturbed-projection}
\end{equation}

Subtract
$\Pi_{\mathcal C_i}\!\Big(
\boldsymbol\xi_i^{k+1}
+\nabla_{\boldsymbol\xi_i}\mathcal L_\rho(\mathbf z^{k+1})
\Big)$
from both sides of~\eqref{eq:fb-perturbed-projection}.
Then
\begin{align}
&\boldsymbol\xi_i^{k+1}
-\Pi_{\mathcal C_i}\!\left(
\boldsymbol\xi_i^{k+1}
+\nabla_{\boldsymbol\xi_i}\mathcal L_\rho(\mathbf z^{k+1})
\right)
\notag\\
&\quad=
\Pi_{\mathcal C_i}\!\left(
\boldsymbol\xi_i^{k+1}
+\nabla_{\boldsymbol\xi_i}\mathcal L_\rho(\mathbf z^{k+1})
+\mathbf e_i
\right)
\notag\\
&\qquad-
\Pi_{\mathcal C_i}\!\left(
\boldsymbol\xi_i^{k+1}
+\nabla_{\boldsymbol\xi_i}\mathcal L_\rho(\mathbf z^{k+1})
\right).
\label{eq:fb-projection-difference}
\end{align}
The two projection inputs differ only by $\mathbf e_i$.
Euclidean projection onto a closed convex set does
not increase distances. Thus, for $i=1,\ldots,4$,
\begin{align}
&\left\|
\boldsymbol\xi_i^{k+1}
-\Pi_{\mathcal C_i}\!\left(
\boldsymbol\xi_i^{k+1}
+\nabla_{\boldsymbol\xi_i}\mathcal L_\rho(\mathbf z^{k+1})
\right)
\right\|\le\|\mathbf e_i\|.
\label{eq:fb-stationarity-bound}
\end{align}

\subsubsection{Step 3: Use the Dual Update
\eqref{eq:fb-dual-update} to Bound the Complete KKT Residual}
The vector $Q\boldsymbol\xi^{k+1}$ contains the
residuals of the equalities \eqref{eq:P1_flow}, \eqref{eq:P1_rate}, \eqref{eq:p1-snragg}, and
\eqref{eq:consensus}, as written in~\eqref{eq:fb-equality-residual}
and~\eqref{eq:fb-equalities}.
The dual update~\eqref{eq:fb-dual-update} expresses
these residuals through the dual iterate change:
\begin{equation}
Q\boldsymbol\xi^{k+1}
=\frac{1}{\rho\tau}\mathbf d_\Lambda.
\label{eq:fb-equality-step}
\end{equation}

The definition~\eqref{eq:fb-kkt-residual} and the
bound~\eqref{eq:fb-stationarity-bound} for each
$i=1,\ldots,4$ give
\begin{align}
&\|\mathbf R(\mathbf z^{k+1})\|^2
\notag\\
&=
\sum_{i=1}^{4}
\left\|
\boldsymbol\xi_i^{k+1}
-\Pi_{\mathcal C_i}\!\left(
\boldsymbol\xi_i^{k+1}
+\nabla_{\boldsymbol\xi_i}\mathcal L_\rho(\mathbf z^{k+1})
\right)
\right\|^2
+\|Q\boldsymbol\xi^{k+1}\|^2
\notag\\
&\le
\sum_{i=1}^{4}\|\mathbf e_i\|^2
+\|Q\boldsymbol\xi^{k+1}\|^2.
\label{eq:fb-complete-residual-bound}
\end{align}

To express the right-hand side in terms of the
iterate changes, define
\begin{equation}
M_R\triangleq
\begin{bmatrix}
\rho Q^{\mathsf T}Q-G_\xi&Q^{\mathsf T}\\
0&\dfrac{1}{\rho\tau}I
\end{bmatrix},
\label{eq:fb-residual-matrix}
\end{equation}
where $G_\xi$ is the primal matrix in~\eqref{eq:fb-G}.
Stacking~\eqref{eq:fb-perturbation} over the four
blocks and using~\eqref{eq:fb-equality-step} gives
\begin{equation}
\begin{bmatrix}
\mathbf e_1\\
\mathbf e_2\\
\mathbf e_3\\
\mathbf e_4\\
Q\boldsymbol\xi^{k+1}
\end{bmatrix}
=
M_R
\begin{bmatrix}
\mathbf d\\
\mathbf d_\Lambda
\end{bmatrix}.
\label{eq:fb-stacked-residual-bound}
\end{equation}
Since
\[
\begin{bmatrix}
\mathbf d\\
\mathbf d_\Lambda
\end{bmatrix}
=
\mathbf z^{k+1}-\mathbf z^k,
\]
equations~\eqref{eq:fb-complete-residual-bound}
and~\eqref{eq:fb-stacked-residual-bound} imply
\begin{align}
\|\mathbf R(\mathbf z^{k+1})\|
&\le
\|M_R(\mathbf z^{k+1}-\mathbf z^k)\|
\notag\\
&\le
\|M_RG^{-1/2}\|_2
\|\mathbf z^{k+1}-\mathbf z^k\|_G.
\label{eq:fb-residual-step}
\end{align}
Here, $\|M\|_2=\sup_{\|\mathbf v\|=1}\|M\mathbf v\|$.
The last inequality follows by writing
\begin{equation}
M_R(\mathbf z^{k+1}-\mathbf z^k)
=
M_RG^{-1/2}G^{1/2}(\mathbf z^{k+1}-\mathbf z^k).
\end{equation}
Define
\begin{equation}
C_R\triangleq\|M_RG^{-1/2}\|_2.
\label{eq:fb-residual-constant}
\end{equation}
The matrices $M_R$ and $G$ depend only on the fixed
network data and algorithm parameters.
Since $G\succ0$, the constant $C_R$ is finite and
independent of $k$.
It is positive because $M_R$ contains the nonzero
block $(\rho\tau)^{-1}I$.
Substituting~\eqref{eq:fb-residual-constant}
into~\eqref{eq:fb-residual-step} proves Lemma \ref{lem:residual_bound}.
\end{IEEEproof}